\documentclass[aps,pra,twocolumn,superscriptaddress]{revtex4-1}
\usepackage{amsmath,amsfonts,amssymb,amscd,mathtools,amsthm}
\usepackage{dsfont}
\usepackage{amsthm}
\usepackage{graphicx}
\usepackage{dcolumn}
\usepackage{bm}
\usepackage{booktabs}
\usepackage{hyperref}
\usepackage[dvipsnames]{xcolor}
\hypersetup{
    bookmarksnumbered=true, 
    unicode=false, 
    pdfstartview={FitH}, 
    pdftitle={}, 
    pdfauthor={}, 
    pdfsubject={}, 
    pdfcreator={}, 
    pdfproducer={}, 
    pdfkeywords={}, 
    pdfnewwindow=true, 
    colorlinks=true, 
    linkcolor=NavyBlue, 
    citecolor=NavyBlue, 
    filecolor=NavyBlue, 
    urlcolor=NavyBlue 
}
\newcounter{one}
\newcommand{\bra}[1]{\langle #1 |}
\newcommand{\ket}[1]{| #1 \rangle}
\newcommand{\braket}[2]{\langle {#1} | {#2} \rangle}
\newcommand{\brakets}[3]{\langle {#1} | {#2} | {#3} \rangle}
\newcommand{\ketbra}[2]{|{#1}\rangle\!\langle{#2}|}
\newcommand{\dm}[1]{\ketbra{#1}{#1}}

\newcommand{\Tr}[0]{ \mathrm{Tr}}

\newcommand{\tot}[0]{ \mathrm{tot}}

\newcommand{\eq}[1]{\begin{align} #1 \end{align}}
\newtheorem{theorem}{Theorem}
\newtheorem{proposition}[theorem]{Proposition}

\newtheorem{lemma}[theorem]{Lemma}
\newtheorem{corollary}[theorem]{Corollary}
\theoremstyle{definition}
\newtheorem{definition}{Definition}

\DeclareMathOperator{\id}{id}

\newcommand\calB{{\cal B}}
\newcommand\calC{{\cal C}}

\newcommand\calE{{\cal E}}
\newcommand\calF{{\cal F}}

\newcommand\calP{{\cal P}}

\newcommand\calR{{\cal R}}
\newcommand\calS{{\cal S}}

\newcommand\calV{{\cal V}}

\newcommand{\bbF}{\mathbb{F}}

\newcommand{\bbO}{\mathbb{O}}
\newcommand{\bbP}{\mathbb{P}}

\newcommand{\bbR}{\mathbb{R}}

\newcommand{\beq}{\begin{equation}}
\newcommand{\eeq}{\end{equation}}
\def\kb#1{\ket{#1}\bra{#1}}
\newcommand{\bal}{\begin{equation}\begin{aligned}}
\newcommand{\eal}{\end{aligned}\end{equation}}

\newcommand{\sbar}{\;\rule{0pt}{9.5pt}\right|\;}
\newcommand{\lset}{\left\{\left.}
\newcommand{\rset}{\right\}}

\newcommand{\HT}[1]{{\color{black} #1}}

\usepackage[lmargin=.7in,rmargin=.7in,tmargin=.7in,bmargin=1in]{geometry}

\newcommand{\revision}[1]{{{#1}}}

\begin{document}
\title{Gibbs-Preserving Operations Requiring Infinite Amount of Quantum Coherence}
\author{Hiroyasu Tajima}
\thanks{Both authors contributed equally to this work.}
\affiliation{
				Department of Informatics, Faculty of Information Science and Electrical Engineering, Kyushu University, 744 Motooka, Nishi-ku, Fukuoka, 819-0395, Japan
			}
\affiliation{
				Graduate School of Informatics and Engineering, The University of Electro-Communications,1-5-1 Chofugaoka, Chofu, Tokyo 182-8585, Japan
			}
\affiliation{
JST, PRESTO, 4-1-8 Honcho, Kawaguchi, Saitama, 332-0012, Japan
			}
\email{hiroyasu.tajima@inf.kyushu-u.ac.jp}
\author{Ryuji Takagi}
\thanks{Both authors contributed equally to this work.}
\affiliation{Department of Basic Science, The University of Tokyo, 3-8-1 Komaba, Meguro-ku, Tokyo 153-8902, Japan}
\email{ryujitakagi.pat@gmail.com}

\begin{abstract}
Gibbs-preserving operations have been studied as one of the standard free processes in quantum thermodynamics. Although they admit a simple mathematical structure, their operational significance has been unclear due to the potential hidden cost to implement them using an operatioanlly motivated class of operations, such as thermal operations. Here, we show that this hidden cost can be infinite---we present a family of Gibbs-preserving operations that cannot be implemented by thermal operations aided by any finite amount of quantum coherence. Our result implies that there are uncountably many Gibbs-preserving operations that require unbounded thermodynamic resources to implement, raising a question about employing Gibbs-preserving operations as available thermodynamics processes. This finding is a consequence of the general lower bounds we provide for the coherence cost of approximately implementing a certain class of Gibbs-preserving operations with a desired accuracy. We find that our lower bound is almost tight, identifying a quantity---related to the energy change caused by the channel to implement---as a fundamental quantifier characterizing the coherence cost for the approximate implementation of Gibbs-preserving operations.

\end{abstract}

\maketitle

\textit{\textbf{Introduction.---}}
A central question in thermodynamics---and quantum extension thereof---is to formalize feasible state transformations under available thermodynamic operations. 
Recent studies have uncovered that this can effectively be studied by a resource-theoretic approach, which admits a rigorous analytical platform. 
There, one considers a class of operations that are ``freely accessible'' in thermodynamic settings and studies operational consequences, e.g., work extraction, under such operations.
Therefore, the outcome of the analysis can naturally depend on the choice of the accessible operations, and it is crucial to recognize and appreciate the justification and potential drawback of those operations. 

The bare minimum of the thermodynamically free operations is that they should map a thermal Gibbs state to a Gibbs state~\cite{Janzing2000thermodynamic,Lostaglio2019introductory}.
One standard choice for thermodynamic operations, known as \emph{Thermal Operations}~\cite{Janzing2000thermodynamic,Horodecki2009quantum}, is to impose an additional physical restriction, where energy-conserving unitary interacting with an ambient heat bath is only allowed. 
This class is operationally well supported, but at the same time it is often hard to analyze due to this additional structure. 
Another standard approach is to consider all operations that meet the minimum Gibbs-preserving requirement, so-called \emph{Gibbs-preserving Operations}, as available thermodynamic processes.
This rather axiomatic approach benefits from a great mathematical simplification, which allowed for several recent key findings in quantum thermodynamics~\cite{Faist2018fundamental,Faist2019thermodynamic,Buscemi2019information,Wang2019resource,Liu2019one-shot,Regula2020benchmarking,Faist2021thermodynamic,Sagawa2021asymptotic,Shiraishi2021quantum}.

Although these two classes have been flexibly chosen depending on the goal of the study, the precise relation between them has largely been unclear. 
In particular, it is not clear at all whether Gibbs-preserving Operations admit physically reasonable realization with respect to Thermal Operations---if not, the status of Gibbs-preserving Operations as thermodynamic processes would be put into question.
Indeed, it has been known for a while that the set of Gibbs-preserving Operations is \emph{strictly} larger than the set of Thermal Operations~\cite{Faist2015Gibbs-preserving}, making the gap between these two classes worth analyzing. 
In fact, Ref.~\cite{Faist2015Gibbs-preserving} revealed that a key difference between these two maps resides in the capability of creating \emph{quantum coherence}---superposition between energy eigenstates---which is known to serve as a useful thermodynamic resource~\cite{Brandao2013resource,Lostaglio2015quantum,Gour2018quantum,Kwon2018clock}.
Thermal Operations cannot create quantum coherence from incoherent states, but Gibbs-preserving Operations can. 
This demands that to realize Gibbs-preserving Operations with Thermal Operations, one generally needs to aid them with extra quantum coherence. 
Beyond this, not much is known about the implementability of Gibbs-preserving Operations, except for the limited case of trivial Hamiltonian~\cite{Faist2015minimal}.
In particular, it is crucial to clarify whether there is a universally sufficient amount of thermodynamic resources that admits implementation of any Gibbs-preserving operation of a fixed size with Thermal Operations, which would secure a certain level of physical justification of Gibbs-preserving Operations. 

Here, we show that it is not the case. We present a continuous family of Gibbs-preserving Operations that cannot be implemented by any finite amount of quantum coherence.
We provide an explicit way of constructing such Gibbs-preserving Operations, which can be applied to arbitrary dimensional systems and almost arbitrary Hamiltonian.
Interestingly, these operations are not the ones that create coherence---like the one discussed in Ref.~\cite{Faist2015Gibbs-preserving}---but the ones that \emph{detect} coherence. We show that the former can actually be implemented by a finite amount of coherence, showing an intriguing asymmetry between coherence creation and detection in terms of implementation cost. 

We show the phenomenon of infinite coherence cost by obtaining the general lower bounds for the coherence cost required to approximately implement a certain class of Gibbs-preserving Operations, which diverges at the limit of zero implementation error. 
We show that our lower bound is almost tight, where we in turn find that a quantifier introduced in Ref.~\cite{Tajima2022universal}, which is linked to the capability of changing energy, characterizes the optimal coherence cost for certain Gibbs-preserving Operations. 
We also find that an arbitrary (not necessarily Gibbs-preserving) quantum channel can generally be approximately implemented by Thermal Operations with a coherence cost that scales with the error in the same way as the aforementioned lower bound for Gibbs-preserving Operations, showing that some Gibbs-preserving Operations, roughly speaking, belong to the most costly class of quantum operation. 

Our results provide a partial solution to the open problem raised in Ref.~\cite{open_problem} and particularly confirm the existence of thermodynamically infeasible Gibbs-preserving Operations.   
Our finding therefore implies that one needs to interpret the operational power of Gibbs-preserving Operations with extra caution in light of their physical implementability.


\textit{\textbf{Preliminaries.---}}
We begin by introducing relevant settings and frameworks. (See Sec.~\ref{app:background} of the  Supplemental Material~\footnote{See Supplemental Material for details of backgrounds, proofs, and extended discussions, which includes Refs.~\cite{Gour2008resource,Chitambar2019quantum,Brandao2010reversible,Brandao2011one-shot,Regula2019one-shot,YungerHalpern2016microcanonical,Brandao2015second,Marvian2012symmetry,Kudo_Tajima,YT,YT2,Shitara_Tajima,Keyl1999optimal,HP,Datta2009min,Regula2018convex,rockafellar2015convex,Wilming2022correlationsin,Shiraishi2024arbitrary}} for more extensive descriptions.) 
\nocite{Gour2008resource,Chitambar2019quantum,Brandao2010reversible,Brandao2011one-shot,Regula2019one-shot,YungerHalpern2016microcanonical,Brandao2015second,Marvian2012symmetry,Kudo_Tajima,YT,YT2,Shitara_Tajima,Keyl1999optimal,HP,Datta2009min,Regula2018convex,rockafellar2015convex,Wilming2022correlationsin,Shiraishi2024arbitrary}
Throughout this work, we consider a situation where systems are surrounded by a thermal bath with an arbitrary finite inverse temperature $\beta$.
We assume that the specification of a system $X$ always comes with its Hamiltonian $H_X=\sum_{i=1}^{d_X} E_{X,i}\dm{i}$ with dimension $d_X$. 
Then, the thermal Gibbs state in system $X$ is written by $\tau_X=e^{-\beta H_X}/\Tr(e^{\beta H_X})$. 

We consider a quantum channel, i.e., completely-positive trace-preserving (CPTP) map, from a \revision{ finite-dimensional} system $S$ to another \revision{finite-dimensional} system $S'$. 
A central class of quantum channels we consider is the set of Gibbs-preserving Operations. 
As the name suggests, these are the operations that map Gibbs states to Gibbs states. 
Here, we employ a generalized notion of Gibbs-preserving Operations, in which input and output systems can generally be different~\cite{Janzing2000thermodynamic,Renes2014work,Faist2019thermodynamic,Faist2021thermodynamic,Sagawa2021asymptotic}. 
Namely, we call a channel $\Lambda:S\to S'$ Gibbs-preserving if $\Lambda(\tau_{S}) = \tau_{S'}$.

Another class, which is supported by an operational consideration, is the set of Thermal Operations. 
We call a channel $\Lambda:S\to S'$ Thermal Operation if there are environments $E$ and $E'$ such that $S\otimes E = S'\otimes E'$ and a unitary $U$ on the whole system satisfying
\begin{equation}\begin{aligned}
 \Lambda(\rho) = \Tr_{E'} \left(U \rho \otimes \tau_E U^\dagger\right),\quad [U,H_{\rm tot}] = 0
 \label{eq:thermal operations definition}
\end{aligned}\end{equation}
where $H_{\rm tot} = H_S\otimes \mathds{1}_E + \mathds{1}_S \otimes H_E = H_{S'}\otimes \mathds{1}_{E'} + \mathds{1}_{S'}\otimes H_{E'}$ is the total Hamiltionian~\cite{Janzing2000thermodynamic,Horodecki2009quantum,footnote_thermal}.

It is not difficult to see that Thermal Operations are always Gibbs preserving. However, the converse is not true.  
Ref.~\cite{Faist2015Gibbs-preserving} showed this by considering a simple example of a qubit channel $\Lambda:S\to S$ defined by
\begin{equation}\begin{aligned}
 \Lambda(\rho) = \brakets{1}{\rho}{1}\eta + \brakets{0}{\rho}{0}\sigma 
 \label{eq:Faist example}
\end{aligned}\end{equation}
where $\eta$ is some quantum state one can choose, and $\sigma = \brakets{0}{\tau_S}{0}^{-1}(\tau_S - \brakets{1}{\tau_S}{1}\eta)$. 
One can explicitly check that this is Gibbs preserving by definition.
On the other hand, by choosing $\eta$ as a state containing energetic coherence, i.e., off-diagonal term with respect to the energy eigenbasis, one can see that this channel can prepare a coherent state from the state $\dm{1}$, which does not have energetic coherence. 
Since Thermal Operations are not able to create energetic coherence from scratch, one can conclude that such a channel is Gibbs-preserving but not a Thermal Operation. 

This indicates that the key notion to fill the gap between Gibbs-preserving and Thermal Operations is the energetic coherence, and we would like to formalize this quantitatively. 
Formally, we say that a state $\rho$ in $S$ has energetic coherence if $\rho \neq e^{-iH_S t}\rho e^{i H_St}$ for some time $t$, which is equivalent to having a nonzero block off-diagonal element with respect to energy eigenbasis. 
For a quantitative analysis of energetic coherence, we employ quantum Fisher information defined for a state $\rho$ in system $S$ by 
\begin{equation}\begin{aligned}
 \calF(\rho) = 2\sum_{i,j}\frac{(\lambda_i-\lambda_j)^2}{\lambda_i+\lambda_j}|\bra{e_i}H_S\ket{e_j}|^2,
\end{aligned}\end{equation}
where $\{\lambda_i\}_i$ and $\{\ket{e_i}\}_i$ are the sets of eigenvalues and eigenstates of a state $\rho$ such that $\rho=\sum_i\lambda_i\dm{e_i}$.
Quantum Fisher information is a well-known coherence quantifier that comes with a natural operational interpretation~\cite{Yadin2016general,Marvian2022operational,skew_resource,Takagi_skew}.


\textit{\textbf{Fundamental coherence cost.---}}
We investigate how costly it is to implement Gibbs-preserving Operations by analyzing the amount of coherence needed to implement a desired Gibbs-preserving Operation by a Thermal Operation.  
Here, we measure the accuracy of implementation by a channel purified distance~\cite{Gilchrist2005distance}
\begin{equation}\begin{aligned}
 D_F(\Lambda_1,\Lambda_2) \coloneqq \max_\rho D_F(\id\otimes\Lambda_1(\rho),\id\otimes\Lambda_2(\rho))
\end{aligned}\end{equation}
where $D_F(\rho,\sigma)= \sqrt{1-F(\rho,\sigma)^2}$ and $F(\rho,\sigma)=\Tr\sqrt{\sqrt{\rho}\sigma\sqrt{\rho}}$.
We particularly write $\Lambda_1\sim_\epsilon \Lambda_2$ to denote $D_F(\Lambda_1,\Lambda_2)\leq \epsilon$.

The primary quantity we study is the minimum coherence cost for implementing a channel $\Lambda$ with error $\epsilon$ defined by 
\begin{equation}\begin{aligned}
 \calF_c^\epsilon(\Lambda) \coloneqq \min\lset \calF(\eta)\sbar \Lambda(\cdot) \sim_\epsilon \tilde\Lambda(\cdot\otimes \eta),\ \tilde\Lambda\in\bbO_{\rm TO}\rset
\end{aligned}\end{equation}
where $\bbO_{\rm TO}$ is the set of Thermal Operations, and $\eta$ is a state in an arbitrary ancillary system. Namely, we regard the coherence cost as the minimum amount of coherence attributed to an ancillary state that---together with a Thermal Operation---realizes the approximation implementation of the target channel $\Lambda$.

The key idea in evaluating this is to connect our setting to the recent trade-off relation between coherence cost for channel implementation and the degree of reversibility of the channel to implement~\cite{Tajima2022universal}, which was shown to unify, e.g., the Wigner-Araki-Yanase theorems on quantum processes~\cite{Wigner1952,Araki-Yanase1960,OzawaWAY,TN,Kuramochi-Tajima,ET2023,ozawaWAY_CNOT,TSS,Tajima2020coherence,TS} and the Eastin-Knill theorems on quantum error correcting codes \cite{Eastin-Knill,e-EKFaist,e-EKKubica, e-EKZhou,e-EKYang,TS}---see Sec.~\ref{app:background} of the Supplemental Material for details. 
In light of this, we find that the following class of Gibbs-preserving Operations plays a central role. 
\begin{definition}\label{def:pairwise reversible}
We call a Gibbs-preserving Operation $\Lambda$ \emph{pairwise reversible} if there exists a pair $\bbP=\{\rho_1,\rho_2\}$ of orthogonal states, i.e., $\Tr(\rho_1\rho_2)=0$, and a quantum channel $\calR$ such that $\calR\circ\Lambda(\rho_j)=\rho_j$ for $j=1,2$. 
We also call $\bbP$ a \emph{reversible pair} of $\Lambda$.
\end{definition} 

We now introduce a central quantity for characterizing coherence cost. Let $\bbP=\{\rho_1,\rho_2\}$ be a reversible pair for a Gibbs-preserving channel $\Lambda:S\to S'$. 
Then, we define
\begin{equation}\begin{aligned}
 \calC(\Lambda,\bbP)\coloneqq \|\sqrt{\rho_1}(H_S-\Lambda^\dagger(H_{S'}))\sqrt{\rho_2}\|_2
 \label{eq:energy change operator definition}
\end{aligned}\end{equation}
where $\Lambda^\dagger$ is the dual map such that $\Tr(\Lambda^\dagger(A)B)=\Tr(A\,\Lambda(B))$ for arbitrary operators $A$ and $B$, and $\|X\|_2\coloneqq \sqrt{\Tr(X^\dagger X)}$ is the Hilbert-Schmidt norm.
This quantity particularly admits a simpler form for pure-state reversible pair $\bbP=\{\psi_1,\psi_2\}$ as
\begin{equation}\begin{aligned}
 \calC(\Lambda,\bbP)=|\brakets{\psi_1}{H_S-\Lambda^\dagger(H_{S'})}{\psi_2}|.
 \label{eq:energy change operator definition pure states}
\end{aligned}\end{equation}
The quantity $H_S-\Lambda^\dagger(H_{S'})$ is an operator that corresponds to the local energy change in the system, and the forms in \eqref{eq:energy change operator definition} and \eqref{eq:energy change operator definition pure states} indicate that $\calC(\Lambda,\bbP)$ measures the off-diagonal element of this operator with respect to the reversible states. 
More discussions about this quantity can be found in \HT{Sec.~\ref{subsec:trade-off app} of the Supplemental Material (see also Ref.~\cite{Tajima2022universal})}.

We are now in the position to present our first main result, which establishes a universal lower bound for the coherence cost for pairwise reversible Gibbs-preserving Operations. (Proof in Sec.~\ref{app:lower bound} of the Supplemental Material.)

\begin{theorem}\label{thm:lower bound}
Let $\Lambda:S\to S'$ be a pairwise reversible Gibbs-preserving Operation with a reversible pair $\bbP$. Then, 
\begin{equation}\begin{aligned}
\sqrt{\calF^\epsilon_c(\Lambda)}\geq 
\frac{\calC(\Lambda,\bbP)}{\epsilon} -\Delta(H_S) - 3\Delta(H_{S'}),\label{eq:lower bound}
\end{aligned}\end{equation}
where $\Delta(O)$ is the difference between the minimum and maximum eigenvalues of an operator $O$.
\end{theorem}
This particularly establishes a demanding coherence cost in the small error regime. 
\HT{In Sec.~\ref{app:lower bound} of the Supplemental Material, we extend Theorem~\ref{thm:lower bound} to the cases without perfect pairwise reversibility, as well as to the form that does not directly depend on the maximum range $\Delta(H_S)$ or $\Delta(H_{S'})$ of the Hamiltonians, showing the potential of obtaining a similar lower bound applicable to unbounded Hamiltonians in infinite-dimensional systems.
}

Theorem~\ref{thm:lower bound} implies that no pairwise reversible Gibbs-preserving Operation $\Lambda$ with $\calC(\Lambda,\bbP)>0$ can be exactly implemented with a finite amount of coherence cost, as the lower bound diverges as $\epsilon$ approaches 0.
Therefore, the problem of whether cost-diverging Gibbs-preserving Operations exists reduces to whether there exists a pairwise reversible Gibbs-preserving Operation $\Lambda$ and a reversible pair $\bbP$ such that $\calC(\Lambda,\bbP)>0$ at all. 
The following result not only shows the existence of such operations but provides a continuous family of those.

\begin{theorem}\label{thm:sufficient condition reversible}
Let $\tau_{X,i}=\brakets{i}{\tau_X}{i}_X$ be the Gibbs distribution for the Gibbs state for a system $X$ with Hamiltonian $H_X = \sum_i E_{X,i}\dm{i}_X$. Then, if there are integers $i$, $j$, and $i'$ for systems $S$ and $S'$ such that
\begin{equation}\begin{aligned}
 \tau_{S,i} < \tau_{S',i'} < \tau_{S,j},
 \label{eq:condition sufficient}
\end{aligned}\end{equation}
there exists a pairwise reversible Gibbs-preserving Operation $\Lambda:S\to S'$ and a reversible pair $\bbP$ such that $\calC(\Lambda,\bbP)>0$.
\end{theorem}
We prove this in Sec.~\ref{app:reversible} of the Supplemental Material, which also provides an explicit construction of the corresponding pairwise reversible Gibbs-preserving Operation.
Here is an illustrative example encompassed in Theorem~\ref{thm:sufficient condition reversible}. Let $S$ and $S'$ be qubit systems with $H_S=\dm{1}$ and $H_{S'}=0$. 
Since $\tau_{S,0}=1/(1+e^{-\beta})$, $\tau_{S,1}=e^{-\beta}/(1+e^{-\beta})$, and $\tau_{S',0}=\tau_{S',1}=1/2$, these systems satisfy \eqref{eq:condition sufficient} for arbitrary finite temperature.  
The corresponding pairwise reversible Gibbs-preserving Operation $\Lambda:S\to S'$ is
\begin{equation}\begin{aligned}
 \Lambda(\rho) = \brakets{+}{\rho}{+}\dm{0} + \brakets{-}{\rho}{-}\dm{1}
 \label{eq:example qubit}
\end{aligned}\end{equation}
where $\ket{\pm}=\frac{1}{\sqrt{2}}(\ket{0}\pm\ket{1})$ is the maximally coherent state on $S$. 
It is easy to see that this is Gibbs-preserving. 
This is also pairwise reversible with a reversible pair $\bbP=\{\dm{+},\dm{-}\}$ because a recovery channel $\calR(\cdot)=\brakets{0}{\cdot}{0}\dm{+}+\brakets{1}{\cdot}{1}\dm{-}$ satisfies $\calR\circ\Lambda(\dm{\pm})=\dm{\pm}$.
Direct computation also shows that $\calC(\Lambda,\bbP)=\frac{1}{2}>0$.

It is insightful to see the structural difference between the Gibbs-preserving Operations in Eqs.~\eqref{eq:Faist example} and \eqref{eq:example qubit}. 
The one in \eqref{eq:Faist example} can create coherence from an incoherent state input state $\dm{1}$.
On the other hand, the channel in \eqref{eq:example qubit} cannot create coherence at all---in fact, output states are always incoherent for any input states. 
Instead, it can perform a measurement in a coherent basis. 
As we show in Sec.~\ref{App:cost_Faist} of the Supplemental Material, 
the coherent cost for the channel in \eqref{eq:Faist example} is upper bounded by $\calF(\eta)+\calF(\sigma)$, which corresponds to the sum of coherence that can be created by the channel. 
This shows a drastic asymmetry between creation and detection of coherence when it comes to its realization. 

We also remark that Theorem~\ref{thm:sufficient condition reversible}, together with Theorem~\ref{thm:lower bound}, guarantees the existence of a Gibbs-preserving Operation with infinite coherence cost for the case when input and output systems are identical. 
Indeed, whenever the system's Hamiltonian comes with at least three distinct eigenenergies, the condition in Theorem~\ref{thm:sufficient condition reversible} with $S'$ being replaced with $S$ is satisfied. 

Theorems~\ref{thm:lower bound}~and~\ref{thm:sufficient condition reversible} provide an overview of the classification of Gibbs-preserving Operations (Fig.~\ref{fig:Venn}). 
We remark that not all pairwise reversible Gibbs-preserving Operations come with a diverging coherence cost (e.g., identity channel)---Theorem~\ref{thm:lower bound} ensures the infinite cost only when there is a reversible pair $\bbP$ satisfying $\calC(\Lambda,\bbP)>0$.
On the other hand, our results do not rule out the possibility that all cost-diverging Gibbs-preserving Operations are pairwise reversible. 
\revision{ We also show that among the set of Gibbs-preserving Operations, cost-diverging ones are atypical. 
Nevertheless, we stress that many well-structured Gibbs-preserving Operations of interest can still come with diverging or extremely high coherence cost---see Sec.~\ref{app:atypicality} of the Supplemental Material for details. 
}

\begin{figure}[t]
\begin{center}
\includegraphics[width=.45\textwidth]{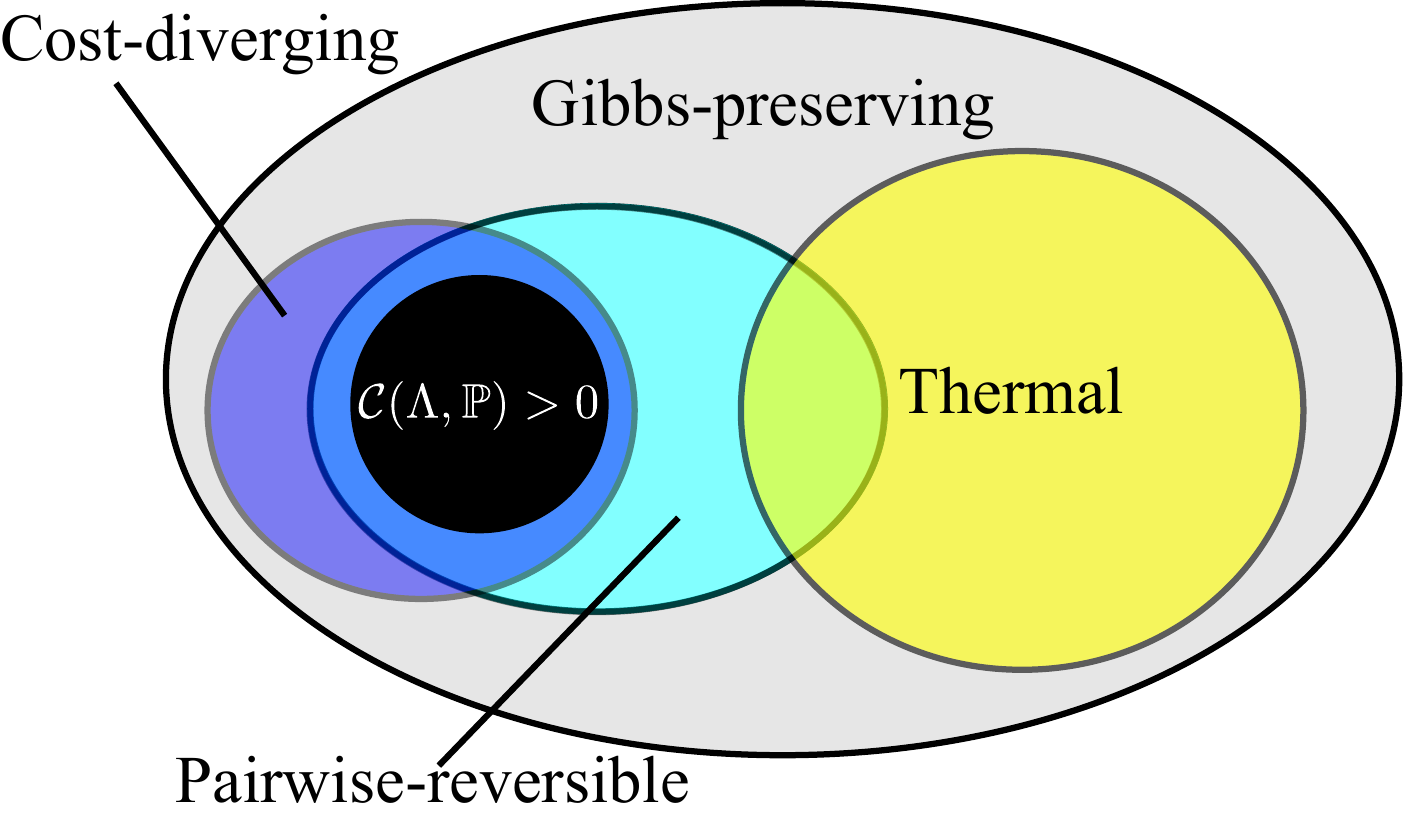}
\caption{Classification of Gibbs-preserving Operations. Theorem~\ref{thm:lower bound} ensures that pairwise reversible Gibbs-preserving Operations (Definition~\ref{def:pairwise reversible}) for which there is a reversible pair $\bbP$ such that $\calC(\Lambda,\bbP)>0$ comes with diverging coherence cost, and the existence of those operations is guaranteed by Theorem~\ref{thm:sufficient condition reversible}. The existence of a cost-diverging Gibbs-preserving channel outside of $\calC(\Lambda,\bbP)>0$ circle has neither been confirmed nor ruled out.}
\label{fig:Venn}
\end{center}
\end{figure}

We furthermore find that our findings restrict the physically feasible state transformations. 
The following result shows that whenever the systems $S$ and $S'$ satisfy the same condition as that for Theorem~\ref{thm:sufficient condition reversible}, there exist state transformations that can be achieved by Gibbs-preserving Operations but none of those Gibbs-preserving Operations cannot be implemented by a finite amount of coherence together with Thermal Operations.

\begin{theorem}\label{thm:restrictions on state transition}
Suppose the input and output systems $S$ and $S'$ satisfy the condition \eqref{eq:condition sufficient}.
Then, there exists a pair $(\rho,\sigma)$ of states such that the state transition $\rho\rightarrow\sigma$ is possible by a Gibbs-preserving Operation but every Gibbs-preserving Operation $\Lambda$ such that $\Lambda(\rho)=\sigma$ must come with diverging coherence cost, i.e., $\calF_c^{\epsilon=0}(\Lambda)=\infty$.
\end{theorem}
We prove this in Sec.~\ref{app:state transition} of the Supplemental Material, which also provides an explicit expression of the pair $(\rho,\sigma)$ of states. 
\revision{ In the proof, we also show that the approximate realization of the cost-diverging state transformation requires the cost inversely proportional to the error.}
Theorem \ref{thm:restrictions on state transition} shows that when the available coherence resource is restricted to be finite, the range of transformations achievable by Gibbs-preserving Operations is strictly narrowed.
We stress that the condition \eqref{eq:condition sufficient} is extremely mild---for virtually all scenarios Theorem~\ref{thm:restrictions on state transition} puts nontrivial restrictions on the physically available state transformations by Gibbs-preserving Operations.

\textbf{\textit{Upper bounds.---}}
A natural next question is how good the bound in Theorem~\ref{thm:lower bound} can be. 
Interestingly, we find that it is almost tight in the following sense.
\begin{theorem}\label{thm:tight bound}
For every real number $a>0$, there is a pairwise reversible Gibbs-preserving Operation $\Lambda$ and a reversible pair $\bbP$ such that $\calC(\Lambda,\bbP)>0$ and 
\begin{equation}\begin{aligned}
\frac{\calC(\Lambda,\bbP)}{\epsilon}-a\leq\sqrt{\calF_c^\epsilon(\Lambda)}\le\frac{\sqrt{2}\calC(\Lambda,\bbP)}{\epsilon}+a.
\end{aligned}\end{equation}
\end{theorem}
Proof can be found in Sec.~\ref{app:upper bounds} of the Supplemental Material.
This particularly ensures that the quantity $\calC(\Lambda,\bbP)$ defined in \eqref{eq:energy change operator definition} serves as a key quantity that characterizes the coherence cost for a certain class of Gibbs-preserving Operations, which in turn provides an operational interpretation to this quantity.  

Theorem~\ref{thm:tight bound} ensures the existence of a Gibbs-preserving Operation that almost achieves the lower bound. However, it does not tell much about the general upper bound that could be applied to an arbitrary Gibbs-preserving Operation.
In the following, we show that, by giving up obtaining the form that almost matches the lower bound, we can obtain the general sufficient coherence cost that can be universally applied to all Gibbs-preserving Operations.
In fact, we find that the applicability of our bound is much beyond Gibbs-preserving Operations---it gives a sufficient coherence cost for an arbitrary quantum channel. 
(Proof in Sec.~\ref{app:upper bounds} of the Supplemental Material.)

\begin{theorem}\label{thm:general upper}
Let $\Lambda:S\to S'$ be an arbitrary quantum channel admitting a dilation form
\begin{equation}\begin{aligned}
 \Lambda(\rho) = \Tr_{E'}\left(V(\rho\otimes \dm{\eta})V^\dagger\right)
\end{aligned}\end{equation}
for some environments $E$ and $E'$ such that $S\otimes E= S'\otimes E'$, some unitary $V$ on $S\otimes E$, and some pure incoherent state \HT{(i.e. an energy eigenstate)} $\ket{\eta}$ on $E$.
Then, 
\begin{equation}\begin{aligned}
\sqrt{\calF_c^\epsilon(\Lambda)}\leq \frac{\Delta(H_{\rm tot}-V^\dagger H_{\rm tot}V)}{2\epsilon} + \sqrt{2}\Delta(H_{\rm tot})
\end{aligned}\end{equation}
where $H_{\rm tot}=H_S\otimes\mathds{1}_E + \mathds{1}_S\otimes H_E=H_{S'}\otimes\mathds{1}_{E'} + \mathds{1}_{S'}\otimes H_{E'}$, and $\Delta(O)$ is the difference between the minimum and maximum eigenvalues of an operator $O$.
\end{theorem}

Notably, this upper bound also scales as $\sim 1/\epsilon$ with the implementation error, which coincides with the asymptotic scaling of the lower bound in Theorem~\ref{thm:lower bound} for pairwise reversible Gibbs-preserving Operations. 
Combining these two, we can understand that the optimal coherence cost for all pairwise reversible Gibbs-preserving Operations with $\calC(\Lambda,\bbP)>0$ are roughly characterized by $N/\epsilon$ where $N$ is an extensive quantity that grows with a particle number of the system. 
This, together with the fact that Theorem~\ref{thm:general upper} applies to an arbitrary quantum channel, also implies that pairwise reversible Gibbs-preserving Operations are as costly as general quantum operations, putting them into the ``most costly'' class to implement.

Let us now remark a unique characteristic of our results in relation to the previous result for trivial Hamiltonian. 
When input and output states are only equipped with trivial Hamiltonian, every state becomes an incoherent state, i.e., invariant under Hamiltonian evolution, and thus coherence loses its status as a precious resource. 
Therefore, the meaningful question in such a setting is to ask the work cost (the minimum number of work bit required) for implementing unital channels using Noisy Operations~\cite{Horodecki2003reversible}, which respectively corresponds to the Gibbs-preserving and Thermal Operations for trivial Hamiltonian. 
Ref.~\cite{Faist2015minimal} showed that the minimum work cost is \emph{upper} bounded by a quantity scaling as $\sim\log(1/\epsilon)$ with implementation error $\epsilon$, which particularly diverges at the limit of exact implementation.
Nevertheless, its lower bound has still not been established, and therefore it is still unclear if this diverging cost is a fundamental phenomenon or merely an artifact of their specific construction, which is based on a decoupling technique~\cite{del_Rio2011thermodynamic}.
On the other hand, our Theorem~\ref{thm:lower bound} provides a \emph{lower} bound that scale with $1/\epsilon$, which is complemented by upper bounds in Theorems~\ref{thm:tight bound}~and~\ref{thm:general upper} with the same scaling.
To the best of our knowledge, our results are the first ones that establish the inherently diverging thermodynamic cost for implementing Gibbs-preserving Operations.

\textit{\textbf{Conclusions.---}}
We established bounds for the minimum coherence cost for implementing Gibbs-preserving Operations with a desired target error.
A major consequence of them is that there are Gibbs-preserving Operations that cannot be implemented with a Thermal Operation aided by any finite amount of quantum coherence, and the approximate implementation of these Gibbs-preserving Operations requires roughly the same amount of coherence that suffices to implement the most costly class of quantum channels. 
Our results therefore clarify an enormous hidden cost in Gibbs-preserving Operations, indicating their unphysical nature as thermodynamically available operations. 

This particularly motivates us to revisit and scrutinize the prior results based on Gibbs-preserving Operations from a physical and operational perspective.  
Indeed, optimal Gibbs-preserving Operations in the standard task of state transformation, e.g., work extraction, are often found to possess a measure-and-prepare structure, which takes a similar form to the one in \eqref{eq:example qubit}.
A close investigation of the coherence cost for such channels will make an important future work. 
Another potential direction is to obtain a finer characterization of thermodynamic costs for Gibbs-preserving Operations. 
This includes the complete tight characterization of the coherence cost for all Gibbs-preserving Operations, as well as obtaining corresponding evaluation for work cost---complementary thermodynamic resource besides quantum coherence.

\let\oldaddcontentsline\addcontentsline
\renewcommand{\addcontentsline}[3]{}

\begin{acknowledgments}
We thank Kaito Watanabe and Bartosz Regula for fruitful discussions. H.T. was supported by JSPS Grants-in-Aid for Scientific Research No. JP19K14610, No. JP22H05250, and
No. JP25K00924, and MEXT KAKENHI Grant-in-Aid for Transformative
Research Areas B ``Quantum Energy Innovation” Grant Numbers 24H00830 and 24H00831, JST PRESTO No. JPMJPR2014, JST MOONSHOT No. JPMJMS2061, and JST FOREST No. JPMJFR2365.
R.T. acknowledges the support of JSPS KAKENHI Grant Number JP23K19028, JP24K16975, JP25K00924, JST, CREST Grant Number JPMJCR23I3, Japan, and MEXT KAKENHI Grant-in-Aid for Transformative
Research Areas A ``Extreme Universe” Grant Number JP24H00943.
\end{acknowledgments}

\let\addcontentsline\oldaddcontentsline

\let\oldaddcontentsline\addcontentsline
\renewcommand{\addcontentsline}[3]{}

\bibliographystyle{apsrmp4-2}
\bibliography{myref}

\begin{thebibliography}{64}
\bibitem [{foo()}]{footnote_thermal}%
  \BibitemOpen
  \href@noop {} {}\bibinfo {note} {In this manuscript, we call a closure of the
  set of Thermal Operations simply Thermal Operations, to which all of our
  results equally apply.}\BibitemShut {Stop}%
\end{thebibliography}

\let\addcontentsline\oldaddcontentsline

\clearpage
\newgeometry{hmargin=1.2in,vmargin=0.8in}

\widetext


\setcounter{theorem}{0}
\renewcommand{\thetheorem}{S.\arabic{theorem}}
\setcounter{figure}{0}
\renewcommand{\thefigure}{S.\arabic{figure}}
\renewcommand{\theequation}{S.\arabic{equation}}

\begin{center}
 {\Large\bf Supplemental Material}
\end{center}

\tableofcontents

\section{Background and setting}\label{app:background}

\subsection{Quantum resource theories} \label{app:general channel cost}

Quantum resource theories~\cite{Gour2008resource,Chitambar2019quantum} provide a useful platform on which quantitative analysis of the underlying quantum features can be performed. 
The core idea of resource theories is to consider relevant sets of quantum states (called \emph{free states}) and operations (called \emph{free operations}) that are easily accessible in the given physical setting.
This results in a framework where one can quantify the amount of precious resources attributed to a given state with respect to the set of free states and investigate the feasible state transformations that can be realized by free operations.

Different physical settings of interest can be specified by appropriately choosing the sets of free states and operations, which leads to different resource theories. 
One standard example is to consider the set of separable states and local operations and classical communication (LOCC), resulting in an operational framework of studying quantum entanglement~\cite{Horodecki2009quantum}.
Here, we remark that for a given set of free states, one can consider choosing a different set of free operations while keeping the essence of the physical setting represented by the necessary requirement for free operations 
\bal
\Lambda(\sigma)\in\bbF,\enskip \forall \Lambda\in\bbO, \enskip \HT{\forall \sigma\in\bbF,}
\label{eq:free operations condition}
\eal
for a set $\bbF$ of free states and a set $\bbO$ of free operations. 
Employing the flexibility in choosing different sets of free operations is usually effective when $\bbO$ comes with a complicated structure and is difficult to analyze.
For instance, in the case of entanglement, it is notoriously hard to study the full potential of LOCC, and therefore several classes of other operations, which include LOCC as their subset, were investigated.
One such set is separability-preserving operations, which is the maximal set that satisfies \eqref{eq:free operations condition}.
This admits a great simplification of the analysis and results in significant insights into entanglement transformation~\cite{Brandao2010reversible,Brandao2011one-shot,Regula2019one-shot}. 

The two examples of free operations in the entanglement theory mentioned above have different perspectives and focuses. 
In particular, LOCC is motivated by an \emph{operational} viewpoint, which aims to reflect the reasonable operations that two distant parties can actually accomplish, while separability-preserving operations employ an \emph{axiomatic} approach that respects the bear minimum constraint that operations should not create entanglement for free. 
It is clear from the definition that the latter contains the former, and the inclusion is indeed strict~\cite{Regula2019one-shot}. 
Although each choice is able to extract different aspects of underlying quantum resources, it is still important to clarify the relationship between them.  
In particular, one crucial question here is how much resources are needed for a free operation in the smaller set $\bbO_1$ to simulate the action of a free operation in the larger set $\bbO_2$, which provides the idea of how ``operationally reasonable'' the axiomatic free operations are.  
One can formalize this by asking a \emph{channel implementation cost} for a channel $\Lambda\in\bbO_1$ defined by 
\bal
 C_{R_\bbF}(\Lambda) \coloneqq \min\lset R_\bbF(\eta) \sbar \Lambda(\cdot)=\tilde\Lambda(\cdot\otimes\eta),\  \tilde\Lambda\in\bbO_2\rset 
\eal
for some resource quantifier $R_\bbF$ with respect to the set $\bbF$ of free states. 
In this work, we study this question in the setting of quantum thermodynamics.


\subsection{Thermal and Gibbs-preserving Operations}

One of the major approaches in quantum thermodynamics is to employ a resource-theoretic framework~\cite{Janzing2000thermodynamic,Horodecki2009quantum,Brandao2013resource}, which focuses on the ultimate operational capability of thermodynamic operations in the manipulation of quantum systems. 
In this approach, the thermal Gibbs state $\tau_S=e^{-\beta H_S}/\Tr(e^{-\beta H_S})$ for a system $S$ with Hamiltonian $H_S$ is considered to be a state that is freely accessible, and allowed thermodynamic operations are chosen so that they map a Gibbs state to another Gibbs state. 

One of the standard choices for such thermodynamic operations is based on an operational motivation and is known as \emph{Thermal Operations}~\cite{Horodecki2009quantum,YungerHalpern2016microcanonical,Lostaglio2019introductory}.
A channel $\Lambda:S\to S'$ is called a Thermal Operation if there exists an environment $E$ and $E'$ with $S\otimes E = S'\otimes E'$ and an energy-conserving unitary $U$ on $S\otimes E$ such that
\bal
 \Lambda(\rho) = \Tr_{E'} \left(U \rho \otimes \tau_S U^\dagger\right),\quad [U,H_{\rm tot}] = 0
 \label{eq:thermal operations definition app}
\eal
where $H_{\rm tot}=H_S\otimes\mathds{1}_E + \mathds{1}_S\otimes H_E = H_{S'}\otimes\mathds{1}_{E'} + \mathds{1}_{S'}\otimes H_{E'}$ is the total Hamiltonian.
In this manuscript, we call a closure of the set of Thermal Operations simply Thermal Operations, to which all of our results equally apply. 
On the other hand, one can also consider a broader class based on the axiomatic formulation, which is only restricted by the minimum requirement that they should map Gibbs states to Gibbs states. 
This is known as \emph{Gibbs-preserving Operations}---a channel $\Lambda:S\to S'$ is called a Gibbs-preserving operation if
\bal
 \Lambda(\tau_{S}) = \tau_{S'}.
\eal

We remark that we here allow the final system $S'$ to differ from the initial system $S$. 
This setting is motivated by the observation that discarding an arbitrary subsystem should operationally be allowed, which may result in a different system from the initial one~\cite{Janzing2000thermodynamic,Horodecki2009quantum,Renes2014work,YungerHalpern2016microcanonical}. 
One could also formulate the change in the initial and final systems by inducing different Hamiltonians by introducing ancillary systems working as a switch~\cite{Horodecki2009quantum,Brandao2015second,Sagawa2021asymptotic}, which generally comes with a different work cost to realize state transformation. 
Although these distinctions should carefully be taken into account when one investigates the work cost, here we do not delve into this discussion further, as our main focus here is the coherence cost, which is not affected by these subtleties. 
We also note that the main consequences of our results, i.e., diverging coherence cost and its bounds for approximate implementation, still hold for the restricted settings with identical initial and final systems, particularly because of the broad applicability of Theorem~\ref{thm:sufficient condition reversible}. See the main text for relevant discussions.

It is elementary to see that any thermal operation is Gibbs-preserving, and therefore the set $\bbO_{\rm TO}$ of thermal operations and the set $\bbO_{\rm GP}$ of Gibbs-preserving operations satisfy the inclusion relation $\bbO_{\rm TO}\subseteq \bbO_{\rm GP}$.
Furthermore, this inclusion is shown to be strict, i.e., $\bbO_{\rm TO}\subsetneq \bbO_{\rm GP}$~\cite{Faist2015Gibbs-preserving}. 
A key observation to see this strict inclusion relation is to study energetic coherence, which we review in the following.


\subsection{Quantum coherence in energy eigenbasis}

Recent studies found that, in the realm of quantum thermodynamics, quantum coherence also plays a major role that allows one to extract work and thus serves as another type of quantum resource besides out-of-equilibrium energy distribution~\cite{Brandao2013resource,Lostaglio2015quantum,Gour2018quantum,Kwon2018clock}. 
To formalize quantum coherence, let $H=\sum_n E_n \dm{n}$ be a Hamiltonian where $\{E_n\}_n$ is the set of (possibly degenerate) energy eigenvalues and $\{\ket{n}\}_n$ is the orthonormal set of energy eigenstates.  
We say that a state $\rho$ has \emph{quantum coherence} or \emph{energetic coherence} if $\rho$ has an off-diagonal term, i.e., superposition, for different energy levels. 
Formally, let $\Pi_E$ be a projector onto the subspace with energy $E$ given by 
\bal
\Pi_{E} = \sum_{n:E_n=E} \dm{n}.
\label{eq:energy projector}
\eal
Then, a state $\rho$ has energetic coherence if $\sum_E \Pi_E\rho \Pi_E\neq \rho$.
\HT{ Equivalently, a state $\rho$ has nonzero coherence if there exists $t\in \bbR$ such that $e^{-H t}\rho e^{iHt} \neq \rho$.}

The letter expression particularly allows us to formalize energetic coherence in relation to a group action represented by a unitary representation of ${\rm U}(1)$ (or $\bbR$ if the Hamiltonian contains relatively irrational eigenvalues). 
Namely, \emph{incoherent states}---states that do not have coherence---are equivalent to the states invariant under a unitary representation $\{e^{-iHt}\}_t$. 
This observation provides a way of quantifying the amount of energetic coherence employing the resource theory of asymmetry~\cite{Gour2008resource,Marvian2012symmetry,skew_resource,Takagi_skew,Kudo_Tajima,YT,YT2,Shitara_Tajima}, which considers states invariant under action of a unitary representation of a group $G$ as free states, i.e., $\bbF=\{\sigma\,|\,U_g \sigma U_g^\dagger = \sigma,\ \forall g\in G\}$, and the operations covariant with such group actions as free operations, i.e., $\bbO = \{\Lambda:S\to S'\,|\, U_{S',g}\Lambda(\rho)U_{S',g}^\dagger = \Lambda(U_{S,g}\rho U_{S,g}^\dagger),\ \forall g\in G\}$ where $U_{X,g}$ is a unitary representation acting on a system $X$. 
One can then realize that incoherent states coincide with the set of free states in the framework of resource theory of asymmetry with U(1) group with the representation $U_t = e^{-iHt}$, equipped with 
\bal
\bbF_{\rm inc}=\{\sigma\,|\,e^{-iH t} \sigma e^{iHt}= \sigma,\ \forall t\}
\eal
and free operation called \emph{Covariant Operations} 
\bal 
\bbO_{\rm cov} = \{\Lambda:S\to S'\,|\, e^{-iH_{S'}t}\Lambda(\rho)e^{iH_{S'}t}= \Lambda(e^{-iH_S t}\rho e^{iH_S t})\ \forall t\}.
\eal
It is easy to check that covariant operations are indeed free operations for energetic coherence in the sense that it does not create a state with nonzero energetic coherence from an incoherent state. 

Another insightful characterization of covariant operations is that an arbitrary covariant operation $\Lambda:S\to S'$ admits the following dilation form~\cite{Keyl1999optimal,Marvian2012symmetry}
\bal
 \Lambda(\rho) = \Tr_{E'} \left(U \rho \otimes \sigma U^\dagger\right),\quad [U,H_S\otimes \mathds{1}_E + \mathds{1}_S\otimes H_E]=0,\ \sigma\in\bbF_{\rm inc}
 \label{eq:covariant operations definition app}
\eal
for some systems $E$ and $E'$.
The similarity between \eqref{eq:thermal operations definition app} and \eqref{eq:covariant operations definition app} represents an intriguing interplay between resource theories for quantum thermodynamics and energetic coherence.
Indeed, since Gibbs states do not have energetic coherence by definition, we immediately notice that 
\bal
\bbO_{\rm TO}\subseteq \bbO_{\rm cov}.
\label{eq:TO in covariant}
\eal
On the other hand, Ref.~\cite{Faist2015Gibbs-preserving} presented a Gibbs-preserving operation $\Lambda\in\bbO_{\rm GP}$ that can create a coherent state from an incoherent state, showing $\Lambda\not\in \bbO_{\rm cov}$. 
This shows the strict inclusion $\bbO_{\rm TO}\subsetneq \bbO_{\rm GP}$.
In other words, in an operationally driven approach with thermal operations, energetic coherence serves as a precious resource that cannot be created for free, while in an axiomatic approach with Gibbs-preserving operations, energetic coherence loses the status of the precious resource.


\subsection{Coherence cost for Gibbs-preserving Operations}

The aforementioned gap between thermal and Gibbs-preserving operations naturally raises a question~\cite{open_problem}: what is the coherence cost for thermal operations to implement Gibbs-preserving operations? 
Indeed, Gibbs-preserving operations have been widely studied because of their simple mathematical structure~\cite{Faist2018fundamental,Faist2019thermodynamic,Buscemi2019information,Wang2019resource,Liu2019one-shot,Regula2020benchmarking,Faist2021thermodynamic,Sagawa2021asymptotic,Shiraishi2021quantum}.
However, if it requires unreasonable additional resource costs to implement, it would lose the physical ground as a reasonable set of free operations from an operational perspective.  

This motivates us to study the channel implementation cost introduced in Sec.~\ref{app:general channel cost} in our setting, which corresponds to $\bbF = \bbF_{\rm inc}$, $\bbO_1=\bbO_{\rm GP}$, and $\bbO_2=\bbO_{\rm TO}$. 
For a coherence quantifier $R_{\bbF_{\rm inc}}$, we employ quantum Fisher information, the standard measure of coherence (and asymmetry in general), defined for an arbitrary state $\rho$ by 
\bal
 \calF(\rho) = 2\sum_{i,j}\frac{(p_i-p_j)^2}{p_i+p_j}|\bra{i}H\ket{j}|^2
\eal
where $H$ is the Hamiltonian of the system that $\rho$ acts on, and $\{p_j\}_j$ and $\{\ket{j}\}_j$ are the eigenvalues and eigenstates of $\rho=\sum_j p_j \dm{j}$.
The coherence cost of a channel $\Lambda$ is then 
\bal
 \calF_c(\Lambda) \coloneqq \min\lset \calF(\eta)\sbar \Lambda = \tilde\Lambda(\cdot\otimes \eta),\ \tilde\Lambda\in\bbO_{\rm TO}\rset.
 \label{eq:coherence cost definition}
\eal

To encompass general and practical scenarios, we extend this quantity to the cost for approximate implementation, admitting some error $\epsilon$.
Here, we measure the error by a channel purified distance~\cite{Gilchrist2005distance}
\bal
 D_F(\Lambda_1,\Lambda_2) \coloneqq \max_\rho D_F(\id\otimes\Lambda_1(\rho),\id\otimes\Lambda_2(\rho))
 \label{eq:purified channel distance definition app}
\eal
where 
\bal
 D_F(\rho,\sigma)= \sqrt{1-F(\rho,\sigma)^2},\quad F(\rho,\sigma)=\Tr\sqrt{\sqrt{\rho}\sigma\sqrt{\rho}}.
\eal
We particularly write $\Lambda_1\sim_\epsilon \Lambda_2$ to denote $D_F(\Lambda_1,\Lambda_2)\leq \epsilon$.
We then define the approximate implementation cost for a channel $\Lambda$ by 
\bal
 \calF_c^\epsilon(\Lambda) \coloneqq \min\lset \calF(\eta)\sbar \Lambda \sim_\epsilon \tilde\Lambda(\cdot\otimes \eta),\ \tilde\Lambda\in\bbO_{\rm TO}\rset.
 \label{eq:approximate coherence cost definition}
\eal


\subsection{Trade-off relation between symmetry, irreversibility, and quantum coherence}\label{subsec:trade-off app}

The main goal of this work is to evaluate the coherence cost in \eqref{eq:approximate coherence cost definition} for Gibbs-preserving Operations. 
In particular, we are interested in the fundamental limitations on the implementation of Gibbs-preserving Operations, which could be analyzed by obtaining lower bounds for $\calF_c^{\epsilon}$.
Lower bounds can never obtained by studying specific implementation protocols, and we thus need an approach that can put general restrictions on all feasible implementation strategies.  

The key technique we employ to this end is the universal trade-off relation between symmetry, irreversibility, and quantum coherence recently found in Ref.~\cite{Tajima2022universal}. 
Let us define a measure of the irreversibility of a channel $\Lambda:S\rightarrow S'$ for an orthogonal state pair $\bbP:=\{\rho_1,\rho_2\}$ on $S$ satisfying $F(\rho_1,\rho_2)=0$ as follows:
\eq{
\delta(\Lambda,\bbP):=\min_{\calR:S'\to S}\sqrt{\sum^{2}_{j=1}\frac{1}{2}D_F(\rho_j,\calR\circ\Lambda(\rho_j))^2}.\label{eq:approximation error}
}
Here $\calR$ runs over all CPTP maps from $S'$ to $S$.
This irreversibility measure gives lower bounds for other relevant quantities, e.g., the entropy production and recovery errors of error-correcting codes~\cite{Tajima2022universal}, and various errors and disturbances of quantum measurements and the out-of-time-ordered correlators (OTOC)~\cite{ET2023}.

For an arbitrary orthogonal state pair $\mathbb{P}=\{\rho_1,\rho_2\}$, we define
\eq{
\calC(\Lambda,\bbP):=\|\sqrt{\rho_1}(H_S-\Lambda^\dagger(H_{S'}))\sqrt{\rho_2}\|_2,
\label{C_supp}
}
where $\|O\|_2=\sqrt{\Tr(O^\dagger O)}$ is the Hilbert-Schmidt norm. 
\HT{We here remark that $\calC(\Lambda,\bbP)$ is well-defined even when $\bbP$ is not a perfectly reversible pair for $\Lambda$. The quantity $\calC(\Lambda,\bbP)$ is well-defined for an arbitrary CPTP map $\Lambda$ and  an arbitrary orthogonal state pair $\mathbb{P}=\{\rho_1,\rho_2\}$.}

We also define coherence cost for channel implementation by covariant operations
\bal
 \calF_{c,{\rm cov}}(\Lambda) \coloneqq \min\lset \calF(\eta)\sbar \Lambda = \tilde\Lambda(\cdot\otimes \eta),\ \tilde\Lambda\in\bbO_{\rm cov}\rset.
 \label{eq:coherence cost covariant definition}
\eal
Then, it turns out that there is a fundamental trade-off relation between these quantities and the coherence cost for exact channel implementation. 

\begin{theorem}[in Ref. \cite{Tajima2022universal}]\label{thm:SIQ}
For an arbitrary quantum channel $\Lambda:S\to S'$ and an arbitrary orthogonal state pair $\mathbb{P}$, the following inequality holds:
\eq{
\frac{\calC(\Lambda,\mathbb{P})}{\sqrt{\calF_{c,{\rm cov}}(\Lambda)}+\Delta(H_S)+\Delta(H_{S'})}\le\delta(\Lambda,\mathbb{P}),\label{eq:SIQ}
} 
where $\Delta(O)$ is the difference between the minimum and maximum eigenvalues of an operator $O$.
\end{theorem}

We remark that the above relation can be extended from $\mathbb{P}$ to a general state ensemble $\Omega:=\{p_j,\rho_j\}_j$---where $\{p_j\}_j$ is probability distribution and $\{\rho_j\}_j$ are states---that may not be orthogonal to each other~\cite{Tajima2022universal}. 
The relation \eqref{eq:SIQ} is given as a unification between the Wigner-Araki-Yanase (WAY) theorems on quantum measurements \cite{Wigner1952,Araki-Yanase1960,OzawaWAY,TN,Kuramochi-Tajima} and unitary gates \cite{ozawaWAY_CNOT,TSS,Tajima2020coherence,TS} and the Eastin-Knill theorems on quantum error correcting codes~\cite{Eastin-Knill,e-EKFaist,e-EKKubica, e-EKZhou,e-EKYang,TS}.
It also allows restrictions on the classical information recovery in the Hayden-Preskill thought experiments~\cite{HP} imposed by the energy conservation \cite{Tajima2022universal} and extends the WAY theorem to various errors and disturbances of quantum measurements and the out-of-time-ordered correlators~\cite{ET2023}.
Here, we utilize this relation to obtain a lower bound for coherence cost for approximately implementing Gibbs-preserving Operations by Thermal Operations. 

\HT{
Next, we discuss the physical meaning of $\calC$. 
To provide an intuitive understanding, the following inequality is useful \cite{Tajima2022universal}:
\begin{align}
\left(\min_k\lambda^{\min}_{>0}(\rho_k)\right)\frac{\calC_F(\Lambda,\mathbb{P})}{2}\le \calC(\Lambda,\mathbb{P})^2 \le\frac{\calC_F(\Lambda,\mathbb{P})}{2}.\label{relation_CF_and_C}
\end{align}
Here, $\{\rho_k\}_{k=1,2}$ are the states in $\bbP$, and $\lambda^{\min}_{>0}(\rho_k)$ is the  minimum non-zero eigenvalue of $\rho_k$. The quantity $\calC_F(\Lambda,\mathbb{P})$ represents the degree of convexity of the Fisher information with respect to the work operator $Y\coloneqq H_S-\Lambda^\dagger(H_{S'})$ as
\begin{align}
\calC_F(\Lambda,\mathbb{P})&:=\sum_{k=1,2}p_{k}\calF(\rho_k;Y)-\calF\left(\sum_{k=1,2}p_{k}\rho_k;Y\right),\\
p_k&:=\frac{1}{2}\enskip (k=1,2).
\end{align}
Here we made the dependence on the observable $Y$ explicit.
Therefore, $\calC$ essentially characterizes the extent to which the coherence (Fisher information) of $Y$ diminishes when the states $\rho_1$ and $\rho_2$ are probabilistically mixed.
Importantly, when $\mathbb{P}$ consists of a pair of orthogonal pure states, as in the scenarios discussed in this letter, the inequality \eqref{relation_CF_and_C} reduces to $\calC^2=\calC_F/2$---meaning that the quantity $\calC$ is characterized by the convexity of Fisher information. 
We also remark that \eqref{relation_CF_and_C} can be extended to the case involving a general test ensemble $\Omega=\{p_j,\rho_j\}_j$ as long as $\{\rho_j\}_j$ are orthogonal to each other.

The physical meaning of $\calC_F$ (and thus $\calC$) becomes clear when we note that $\calC_F$ is analogous to the Holevo quantity $\chi:=S(\sum_jp_j\rho_j)-\sum_jp_jS(\rho_j)$, with entropy being replaced by Fisher information with respect to $Y$. Holevo’s $\chi$ quantifies the information loss in terms of entropy when states are probabilistically mixed. Consequently, $\calC_F$ can be understood as an indicator of the resource loss in terms of Fisher information with respect to $Y$ when the states in $\mathbb{P}$ are probabilistically mixed.
}


\section{Lower bound for pairwise reversible Gibbs-preserving Operations (Proof of Theorem~\ref{thm:lower bound})}\label{app:lower bound}

\subsection{Proof of Theorem~\ref{thm:lower bound}}

Recall that we call a channel $\Lambda$ pairwise reversible with a reversible pair $\bbP$ if each state in $\bbP$ is perfectly reversible, i.e., $\delta(\Lambda,\bbP)=0$. (See Definition~\ref{def:pairwise reversible} in the main text.)
We then obtain the following lower bound for coherence cost.

\begin{theorem}[Theorem~\ref{thm:lower bound} in the main text]\label{thm:lower bound app}
Let $\Lambda:S\to S'$ be a pairwise reversible Gibbs-preserving Operation with a reversible pair $\bbP$. Then, 
\bal
\sqrt{\calF^\epsilon_c(\Lambda)}\geq 
\frac{\calC(\Lambda,\bbP)}{\epsilon} -\Delta(H_S) - 3\Delta(H_{S'}),
\label{eq:lower bound coherence cost general app}
\eal
where $\Delta(O)$ is the difference between the minimum and maximum eigenvalues of an operator $O$.
\end{theorem}

\begin{proof}
Let $\Lambda_\epsilon$ be a channel that approximates $\Lambda$ with error $\epsilon$, i.e., $D_F(\Lambda_\epsilon,\Lambda)\leq \epsilon$ (recall \eqref{eq:purified channel distance definition app}).
We aim to apply Theorem~\ref{thm:SIQ} to $\Lambda_\epsilon$ while expressing each term by the quantities in \eqref{eq:lower bound coherence cost general app} relevant to the desired channel $\Lambda$ and the accuracy of implementation. 

We first bound $\delta(\Lambda_\epsilon,\bbP)$ by the implementation error $\epsilon$.
Let $\calR$ be a recovery channel such that $\calR\circ\Lambda(\rho)=\rho$ for each state $\rho\in\bbP$, whose existence is ensured by assumption. Then, for each state $\rho$ in $\bbP$, we have 
\bal
 D_F(\calR\circ\Lambda_\epsilon(\rho),\rho)&\leq D_F(\calR\circ\Lambda_\epsilon(\rho),\calR\circ\Lambda(\rho))+D_F(\calR\circ\Lambda(\rho),\rho)\\
 &= D_F(\calR\circ\Lambda_\epsilon(\rho),\calR\circ\Lambda(\rho))\\
 &\leq D_F(\Lambda_\epsilon(\rho),\Lambda(\rho))\\
 &\leq \max_\rho D_F(\id\otimes \Lambda_\epsilon(\rho),\id\otimes\Lambda(\rho))\\
 &\leq \epsilon
\eal
where in the first line we used the triangle inequality of the purified distance~\cite{Gilchrist2005distance}, the second line is because of the perfect reversibility of $\rho$ with $\calR$, the third line follows from the data-processing inequality of the purified distance, and the fifth line is because of the assumption that $\Lambda_\epsilon\sim_\epsilon \Lambda$.
This particularly means that 
\bal
\delta(\Lambda_\epsilon,\bbP)\leq \sqrt{\frac{1}{2}D_F(\calR\circ\Lambda_\epsilon(\rho_1),\rho_1)^2+\frac{1}{2}D_F(\calR\circ\Lambda_\epsilon(\rho_2),\rho_2)^2}\leq \epsilon
\label{eq:irreversibility bound app}.
\eal

We next obtain an expression of $\calC(\Lambda_\epsilon,\bbP)$ in terms of $\calC(\Lambda,\bbP)$.
We first get 
\bal
\calC(\Lambda_\epsilon,\bbP) &= \|\sqrt{\rho_1}(H_S-\Lambda_\epsilon^\dagger(H_{S'}))\sqrt{\rho_2}\|_2\\
&\geq \|\sqrt{\rho_1}(H_S-\Lambda^\dagger(H_{S'}))\sqrt{\rho_2}\|_2 - \|\sqrt{\rho_1}(\Lambda^\dagger(H_{S'})-\Lambda_\epsilon^\dagger(H_{S'}))\sqrt{\rho_2}\|_2\\
&= \calC(\Lambda,\bbP) - \|\sqrt{\rho_1}(\Lambda^\dagger(H_{S'})-\Lambda_\epsilon^\dagger(H_{S'}))\sqrt{\rho_2}\|_2  
\label{eq:bound on approximate C app}
\eal
where in the second line we used the triangle inequality of the Hilbert-Schmidt norm. 
We therefore focus on upper bounding the second term $\|\sqrt{\rho_1}(\Lambda^\dagger(H_{S'})-\Lambda_\epsilon^\dagger(H_{S'}))\sqrt{\rho_2}\|_2$.
For states $\rho_1,\rho_2\in\bbP$, let $\rho_1=\sum_kq_k\psi_{k}$ and $\rho_2=\sum_kq'_k\phi_{k}$ be their spectral decompositions, i.e., $\braket{\psi_{k_1}}{\psi_{k_2}}=\delta_{k_1 k_2}$ and $\braket{\phi_{k'_1}}{\phi_{k'_2}}=\delta_{k'_1 k'_2}$ . Then, direct computation gives
\bal
 \|\sqrt{\rho_1}(\Lambda^\dagger(H_{S'})-\Lambda_\epsilon^\dagger(H_{S'}))\sqrt{\rho_2}\|_2  &= \left\|\sum_{k,k'} \sqrt{q_k}\sqrt{q_k'} \psi_k (\Lambda^\dagger(H_{S'})-\Lambda_\epsilon^\dagger(H_{S'}))\phi_{k'}\right\|_2\\
 & = \sqrt{\sum_{k,k'}q_k q_{k'}\left|\bra{\psi_k} (\Lambda^\dagger(H_S)-\Lambda_\epsilon^\dagger(H_{S'}))\ket{\phi_{k'}}\right|^2}.
 \label{eq:bound on Hilbert Schmidt app}
\eal
We further remark that $\braket{\psi_k}{\phi_{k'}}=0$ for all $k$ and $k'$ because $\Tr(\rho_1\rho_2) = 0$ by the definition of reversible pairs. 
For arbitrary orthogonal pure states $\psi$ and $\phi$, the following relation holds:

\bal
|\bra{\psi}\Lambda_\epsilon^\dagger (H_{S'})-\Lambda^\dagger(H_{S'})\ket{\phi}| &\leq \frac{1}{2}|\bra{\psi}\Lambda_\epsilon^\dagger (H_{S'})-\Lambda^\dagger(H_{S'})\ket{\phi}+\bra{\phi}\Lambda_\epsilon^\dagger (H_{S'})-\Lambda^\dagger(H_{S'})\ket{\psi}|\\
&\quad + \frac{1}{2}|\bra{\psi}\Lambda_\epsilon^\dagger (H_{S'})-\Lambda^\dagger(H_{S'})\ket{\phi}-\bra{\phi}\Lambda_\epsilon^\dagger (H_{S'})-\Lambda^\dagger(H_{S'})\ket{\psi}|\\
&=\frac{1}{2}\left|\Tr(H_{S'}(\Lambda_\epsilon-\Lambda)(\ketbra{\psi}{\phi}+\ketbra{\phi}{\psi}))\right|+\frac{1}{2}\left|\Tr(H_{S'}(\Lambda_\epsilon-\Lambda)(\ketbra{\psi}{\phi}-\ketbra{\phi}{\psi}))\right|\\
&=\frac{1}{2}\left|\Tr((H_{S'}-a\mathds{1}_{S'})(\Lambda_\epsilon-\Lambda)(\ketbra{\psi}{\phi}+\ketbra{\phi}{\psi}))\right|\\
&\quad+\frac{1}{2}\left|\Tr((H_{S'}-a\mathds{1}_{S'})(\Lambda_\epsilon-\Lambda)(\ketbra{\psi}{\phi}-\ketbra{\phi}{\psi}))\right|\\
&\leq \frac{1}{2}\|H_{S'}-a\mathds{1}_{S'}\|_\infty \|(\Lambda_\epsilon-\Lambda)(\ketbra{\psi}{\phi}+\ketbra{\phi}{\psi})\|_1\\
&\quad+\frac{1}{2}\|H_{S'}-a\mathds{1}_{S'}\|_\infty \|(\Lambda_\epsilon-\Lambda)(\ketbra{\psi}{\phi}-\ketbra{\phi}{\psi})\|_1\\
&= \frac{1}{2}\|H_{S'}-a\mathds{1}_{S'}\|_\infty \|(\Lambda_\epsilon-\Lambda)(\dm{\eta_+}-\dm{\eta_-})\|_1\\
&\quad+\frac{1}{2}\|H_{S'}-a\mathds{1}_{S'}\|_\infty \|(\Lambda_\epsilon-\Lambda)(\dm{\eta_+'}-\dm{\eta_-'})\|_1\\
&\leq 4\|H_{S'}-a\mathds{1}_{S'}\|_\infty \epsilon\\
&\leq 2\Delta(H_{S'})\epsilon
\eal
where $a$ is an arbitrary real number, and $\ket{\eta_\pm}\coloneqq\frac{1}{\sqrt{2}}(\ket{\psi}\pm \ket{\phi})$ and $\ket{\eta'_\pm}\coloneqq\frac{1}{\sqrt{2}}(\ket{\psi}\pm i\ket{\phi})$.
In the second last line, we used the triangle inequality for the trace norm and $\frac{1}{2}\|\rho-\sigma\|_1\leq D_F(\rho,\sigma)$. In the last line, we fixed $a$ to satisfy $\|H_{S'}\|_\infty=\Delta(H_{S'})/2$.

Together with \eqref{eq:bound on Hilbert Schmidt app}, this particularly implies 
\bal
  \|\sqrt{\rho_1}(\Lambda^\dagger(H_{S'})-\Lambda_\epsilon^\dagger(H_{S'}))\sqrt{\rho_2}\|_2 & \leq  \sqrt{\sum_{k,k'}q_k q_{k'}}2\Delta(H_{S'})\epsilon \\
  & = 2\Delta(H_{S'})\epsilon. 
\eal
Combining this with \eqref{eq:bound on approximate C app}, we get 
\bal
 \calC(\Lambda_\epsilon,\bbP)\geq \calC(\Lambda,\bbP) - 2\Delta(H_{S'})\epsilon.
 \label{eq:lower bound on approximate C final app}
\eal

We finally note that since $\bbO_{\rm TO}\subseteq \bbO_{\rm cov}$ as in \eqref{eq:TO in covariant}, $\calF_c^\epsilon(\Lambda)\geq \calF_{c,\rm{cov}}^\epsilon(\Lambda)$ always holds. 
We conclude the proof by combining all these observations. 
Let us particularly take $\Lambda_\epsilon$ to be the optimal channel achieving the coherence cost with error $\epsilon$, i.e., $\Lambda_\epsilon\sim_\epsilon \Lambda$ and $\calF_{c,{\rm cov}}(\Lambda_\epsilon)=\calF_{c,{\rm cov}}^\epsilon(\Lambda)$. Then, 
\bal
 \sqrt{\calF_c^\epsilon(\Lambda)}&\geq 
 \sqrt{\calF_{c,{\rm cov}}^\epsilon(\Lambda)}\\
 &= \sqrt{\calF_{c,{\rm cov}}(\Lambda_\epsilon)}\\
 &\geq \frac{\calC(\Lambda_\epsilon,\bbP)}{\delta(\Lambda_\epsilon,\bbP)}-\Delta(H_S) - \Delta(H_{S'})\\
 & \geq \frac{\calC(\Lambda,\bbP)}{\epsilon} - \Delta(H_S) - 3\Delta(H_{S'})
\eal
where we used Theorem~\ref{thm:SIQ} in the third line and \eqref{eq:irreversibility bound app} and \eqref{eq:lower bound on approximate C final app} in the fourth line.

\end{proof}

\HT{

\subsection{Extension to the cases without perfect pairwise reversibility}\label{app:without reversibility}

We here remark that we can extend Theorem \ref{thm:lower bound app} to the case where $\Lambda$ is not pairwise reversible for $\bbP$.
\begin{theorem}\label{thm:lower bound app_ex}
Let $\Lambda:S\to S'$ be a Gibbs-preserving Operation. Then, for an arbitrary orthogonal pair states $\bbP$, the following relation holds:
\bal
\sqrt{\calF^\epsilon_c(\Lambda)}\geq 
\frac{\calC(\Lambda,\bbP)}{\epsilon+\delta(\Lambda,\bbP)} -\Delta(H_S) - 3\Delta(H_{S'}),
\label{eq:lower bound coherence cost general app_ex}
\eal
where $\Delta(O)$ is the difference between the minimum and maximum eigenvalues of an operator $O$.
\end{theorem}

The inequality \eqref{eq:lower bound coherence cost general app_ex} is a generalization of \eqref{eq:lower bound coherence cost general app}. Indeed, when $\Lambda$ is reversible for $\bbP$, \eqref{eq:lower bound coherence cost general app_ex} reduces to \eqref{eq:lower bound coherence cost general app}.

\begin{proof}
Similarly to the proof of Theorem \ref{thm:lower bound app}, let $\Lambda_\epsilon$ be a channel that approximates $\Lambda$ with error $\epsilon$, i.e., $D_F(\Lambda_\epsilon,\Lambda)\leq \epsilon$.
We only have to show 
\begin{align}
\delta(\Lambda_\epsilon,\bbP)\le \delta(\Lambda,\bbP)+\epsilon.\label{eq:error and irrev}
\end{align}
The other parts are the same as the proof of Theorem \ref{thm:lower bound app}.

Let $\calR$ be a recovery channel satisfying the following relation, whose existence is ensured by the definition of $\delta(\Lambda,\bbP)$.
\begin{align}
\delta(\Lambda,\bbP)=\frac{D_F(\calR\circ\Lambda(\psi),\psi)^2+D_F(\calR\circ\Lambda(\phi),\phi)^2}{2}.
\end{align}
Then, we obtain \eqref{eq:error and irrev} as follows:
\begin{align}
\delta(\Lambda_\epsilon,\bbP)^2&\le\frac{D_F(\calR\circ\Lambda(\psi),\psi)^2+D_F(\calR\circ\Lambda(\phi),\phi)^2}{2}\nonumber\\
&\le\frac{(D_F(\calR\circ\Lambda_\epsilon(\psi),\psi)+D_F(\calR\circ\Lambda_\epsilon(\psi),\calR\circ\Lambda(\psi))^2+(D_F(\calR\circ\Lambda_\epsilon(\phi),\phi)+D_F(\calR\circ\Lambda_\epsilon(\phi),\calR\circ\Lambda(\phi))^2}{2}\nonumber\\
&\le\frac{(D_F(\calR\circ\Lambda_\epsilon(\psi),\psi)+\epsilon)^2+(D_F(\calR\circ\Lambda_\epsilon(\phi),\phi)+\epsilon)^2}{2}\nonumber\\
&=\frac{D_F(\calR\circ\Lambda_\epsilon(\psi),\psi)^2+2D_F(\calR\circ\Lambda_\epsilon(\psi),\psi)\epsilon+\epsilon^2+D_F(\calR\circ\Lambda_\epsilon(\phi),\phi)^2+2D_F(\calR\circ\Lambda_\epsilon(\phi),\phi)\epsilon+\epsilon^2}{2}\nonumber\\
&=\delta(\Lambda,\bbP)^2+2\frac{D_F(\calR\circ\Lambda_\epsilon(\psi),\psi)+D_F(\calR\circ\Lambda_\epsilon(\phi),\phi)}{2}\epsilon+\epsilon^2\nonumber\\
&\le\delta(\Lambda,\bbP)^2+2\sqrt{\frac{D_F(\calR\circ\Lambda_\epsilon(\psi),\psi)^2+D_F(\calR\circ\Lambda_\epsilon(\phi),\phi)^2}{2}}\epsilon+\epsilon^2\nonumber\\
&=(\delta(\Lambda,\bbP)+\epsilon)^2.
\end{align}
\end{proof}
}

\HT{
\subsection{Towards the extension to systems with unbounded Hamiltonians}\label{app:pre-unbounded}

So far, the results in this letter assume that $S$ and $S'$ are finite dimensional systems, where $H_S$ and $H_{S'}$ are guaranteed to be bounded operators.
This assumption is crucial for Theorem~\ref{thm:lower bound} to be nontrivial, where the lower bound contains the quantities $\Delta(H_S)$ and $\Delta(H_{S'})$ that diverge for unbounded Hamiltonians. 
(We remark that Theorem \ref{thm:sufficient condition reversible} directly holds in this case as well because the proof does not rely on the spectra of $H_S$ and $H_{S'}$.)

In this section, we provide a preliminary result toward extending Theorem~\ref{thm:lower bound} to the one that provides a nontrivial bound for $H_{S}$ and $H_{S'}$ with unbounded spectra.
The following result shows that the quantities $\Delta(H_S)$ and $\Delta(H_{S'})$ can be replaced by the ones that can take finite values even at the limit of unbounded Hamiltonians. 

\begin{theorem}\label{thm:pre-unbounded}
For finite-dimensional systems $S$ and $S'$, let $\Lambda:S\to S'$ be a pairwise reversible Gibbs-preserving Operation with a reversible pure state pair $\bbP:=\{\dm{\psi},\dm{\phi}\}$. 
Let $\tilde{\Lambda}:S\to S'$ be a CPTP map satisfying $D_F(\Lambda,\tilde{\Lambda})\le\epsilon$ and the following inequalities for a real positive number $\Delta$:
\begin{align}
\max_{\rho\in\mathrm{span}(\bbP)}\left|\Tr\left(H_{S'}\tilde{\Lambda}(\rho)\right)-\Tr\left(H_{S'}\Lambda(\rho)\right)\right|&\le\epsilon \Delta,\label{eq:ex-close}\\
\max_{\rho\in\mathrm{span}(\bbP)}|V(\tilde{\Lambda}(\rho);H_{S'})-V(\Lambda(\rho);H_{S'})|&\le\epsilon \Delta^2,\label{eq:v-close}
\end{align}
where $\mathrm{span}(\mathbb{P})$ refers to the set of quantum states in the linear span of $\mathbb{P}$, and $V(\eta;O)$ denotes the variance of the observable $O$ for the state $\eta$.
Then, the following inequality holds:
\begin{equation}\begin{aligned}
\sqrt{\calF_c(\tilde{\Lambda})}\geq 
\frac{\calC(\Lambda,\bbP)}{\epsilon} -2\max_{\rho\in\mathrm{span}(\bbP)}(\sqrt{V(\rho;H_{S})}+\sqrt{V(\Lambda(\rho);H_{S'})}) - 2(1+\sqrt{\epsilon})\Delta.\label{eq:pre-unbounded}
\end{aligned}\end{equation}
\end{theorem}
Although this result assumes the finite dimensionality of $S$ and $S'$, the inequality \eqref{eq:pre-unbounded} does not directly depend on the maximum range of the spectra of $H_S$ or $H_{S'}$. Therefore, it is expected that if we can extend the results in Ref.~\cite{Tajima2022universal}---which assumes the finite dimensionality---to the infinite-dimensional cases involving unbounded operators, Theorem~\ref{thm:pre-unbounded} could also be extended to unbounded $H_S$ and $H_{S'}$, which provides nontrivial lower bounds for coherence cost and diverges at the limit of $\epsilon\to 0$.
We leave the thorough investigation of this problem as a future work.

\begin{proof}

To prove \eqref{eq:pre-unbounded}, we use the following theorem:
\begin{theorem}[Ref.~\cite{Tajima2022universal}]\label{thm:SIQ another}
For finite-dimensional systems $S$ and $S'$, let $\calE:S\to S'$ and $\mathbb{P}$ be an arbitrary quantum channel and an arbitrary orthogonal state pair.
Let $(U,\rho_E,H_{E},H_{E'})$ be an implementation of $\Lambda$ satisfying
\eq{
U^\dagger(H_{S'}+H_{E'})U&=H_{S}+H_{E}\\
\calE(...)&=\Tr_{E'}[U...\otimes\rho_{E}U^\dagger].
}
Then, it holds that
\eq{
\frac{\calC(\calE,\mathbb{P})}{\sqrt{\calF(\rho_E;H_E)}+\Delta_F}\le\delta(\calE,\mathbb{P}),\label{eq:SIQ another}
} 
where 
\eq{
\Delta_F:=\max_{\rho\in\mathrm{span}(\bbP)}\sqrt{\calF(\rho\otimes\rho_E;H_{S}\otimes1_E-U^\dagger H_{S'}\otimes 1_{E'} U)}
}
with $\calF(\eta;O)$ being the Fisher information of the observable $O$ for the state $\eta$.
\end{theorem}

One can upper bound $\Delta_F$ in Theorem~\ref{thm:SIQ another} by
\eq{
\Delta_F&=\max_{\rho\in\mathrm{span}(\bbP)}\sqrt{\calF(\rho\otimes\rho_E;H_{S}\otimes1_E-U^\dagger H_{S'}\otimes 1_{E'} U)}\nonumber\\
&\stackrel{(a)}{\le}\max_{\rho\in\mathrm{span}(\bbP)}\left(\sqrt{\calF(\rho\otimes\rho_E;H_{S}\otimes1_E)}
+\sqrt{\calF(\rho\otimes\rho_E;U^\dagger H_{S'}\otimes 1_{E'} U)}\right)\nonumber\\
&\stackrel{(b)}{\le}2\max_{\rho\in\mathrm{span}(\bbP)}\left(\sqrt{V(\rho\otimes\rho_E;H_{S}\otimes1_E)}
+\sqrt{V(\rho\otimes\rho_E;U^\dagger H_{S'}\otimes 1_{E'} U)}\right)\nonumber\\
&\stackrel{(c)}{=}2\max_{\rho\in\mathrm{span}(\bbP)}\left(\sqrt{V(\rho;H_{S})}
+\sqrt{V(\calE(\rho);H_{S'})}\right),\label{eq:evaluate Delta_F}
}
where we used $\sqrt{\calF(\eta;O_1+O_2)}\le\sqrt{\calF(\eta;O_1)}+\sqrt{\calF(\eta,O_2)}$ in (a), $\calF(\eta;O)\le4 V(\eta;O)$ in (b) and $V(\rho\otimes\rho_E;U^\dagger H_{S'}\otimes 1_{E'} U)=\Tr\left[U^\dagger H^2_{S'}\otimes 1_{E'}U(\rho\otimes\rho_E)\right]-\Tr\left[U^\dagger H_{S'}\otimes 1_{E'}U(\rho\otimes\rho_E)\right]^2=V(\calE(\rho);H_{S'})$ in (c).

Also, when $(U,\rho_E,H_{E},H_{E'})$ implements a channel $\calE$, we obtain
\eq{
\calF_{c}(\calE)\ge\calF_{c,{\rm cov}}(\calE)\ge\calF(\rho_E;H_{E}).\label{eq:evaluate cost}
}

Combining \eqref{eq:SIQ another}, \eqref{eq:evaluate Delta_F} and \eqref{eq:evaluate cost}, and substituting $\tilde{\Lambda}$ for $\calE$, we obtain
\eq{
\sqrt{\calF_{c}(\tilde{\Lambda})}\ge\frac{\calC(\tilde{\Lambda},\bbP)}{\delta(\tilde{\Lambda},\bbP)}-2\max_{\rho\in\mathrm{span}(\bbP)}\left(\sqrt{V(\rho;H_{S})}
+\sqrt{V(\tilde{\Lambda}(\rho);H_{S'})}\right).
}
Therefore, we only have to evaluate $\delta(\tilde{\Lambda},\bbP)$, $\calC(\tilde{\Lambda},\bbP)$ and $V(\tilde{\Lambda}(\rho);H_{S'})$.

In the same manner as the proof of Theorem \ref{thm:lower bound app}, we obtain
\eq{
\delta(\tilde{\Lambda},\bbP)\le\epsilon.
}
Also, due to \eqref{eq:v-close}, the following is valid for arbitrary $\rho\in\mathrm{span}(\bbP)$:
\eq{
\sqrt{V(\tilde{\Lambda}(\rho);H_{S'})}&\le\sqrt{V(\Lambda(\rho);H_{S'})+\epsilon \Delta^2}\le\sqrt{V(\Lambda(\rho);H_{S'})}+\sqrt{\epsilon} \Delta.
}
Therefore, to obtain \eqref{eq:pre-unbounded}, we only have to show
\eq{
\calC(\tilde{\Lambda},\bbP)\ge\calC(\Lambda,\bbP)-2\epsilon\Delta.\label{eq:evaluate C}
}
In the same manner as \eqref{eq:bound on approximate C app}, we obtain for $\mathbb{P}=\{\dm{\psi},\dm{\phi}\}$ that
\bal
\calC(\tilde\Lambda,\bbP) &\ge \calC(\Lambda,\bbP) - \|\sqrt{\dm{\psi}}(\Lambda^\dagger(H_{S'})-\tilde\Lambda^\dagger(H_{S'}))\sqrt{\dm{\phi}}\|_2\\
&=\calC(\Lambda,\bbP) - |\bra{\psi}\tilde\Lambda^\dagger (H_{S'})-\Lambda^\dagger(H_{S'})\ket{\phi}|
.
\eal
The second term can be upper bounded as
\bal
|\bra{\psi}\tilde\Lambda^\dagger (H_{S'})-\Lambda^\dagger(H_{S'})\ket{\phi}| &\leq \frac{1}{2}|\bra{\psi}\tilde\Lambda^\dagger (H_{S'})-\Lambda^\dagger(H_{S'})\ket{\phi}+\bra{\phi}\tilde\Lambda^\dagger (H_{S'})-\Lambda^\dagger(H_{S'})\ket{\psi}|\\
&\quad + \frac{1}{2}|\bra{\psi}\tilde\Lambda^\dagger (H_{S'})-\Lambda^\dagger(H_{S'})\ket{\phi}-\bra{\phi}\tilde\Lambda^\dagger (H_{S'})-\Lambda^\dagger(H_{S'})\ket{\psi}|\\
&=\frac{1}{2}\left|\Tr(H_{S'}(\tilde\Lambda-\Lambda)(\ketbra{\psi}{\phi}+\ketbra{\phi}{\psi}))\right|+\frac{1}{2}\left|\Tr(H_{S'}(\tilde\Lambda-\Lambda)(\ketbra{\psi}{\phi}-\ketbra{\phi}{\psi}))\right|\\
&=\frac{1}{2}\left|\Tr(H_{S'}(\tilde\Lambda-\Lambda)(\eta_+-\eta_{-}))\right|+\frac{1}{2}\left|\Tr(H_{S'}(\tilde\Lambda-\Lambda)(\eta'_+-\eta'_{-}))\right|\\
&\leq\frac{1}{2}\left(\left|\Tr(H_{S'}(\tilde\Lambda-\Lambda)(\eta_+))\right|+\left|\Tr(H_{S'}(\tilde\Lambda-\Lambda)(\eta_-))\right|\right.\\
&\enskip+\left.\left|\Tr(H_{S'}(\tilde\Lambda-\Lambda)(\eta'_+))\right|+\left|\Tr(H_{S'}(\tilde\Lambda-\Lambda)(\eta'_-))\right|\right)\\
&\leq 2\epsilon\Delta
\eal
where we introduced $\eta_\pm=\dm{\eta_\pm}$ and $\eta_\pm'=\dm{\eta_\pm'}$ with $\ket{\eta_\pm}\coloneqq \frac{1}{\sqrt{2}}(\ket{\psi}\pm\ket{\phi})$ and $\ket{\eta_\pm'}\coloneqq \frac{1}{\sqrt{2}}(\ket{\psi}\pm i\ket{\phi})$.
Therefore, we have obtained \eqref{eq:evaluate C}, showing \eqref{eq:pre-unbounded}.
\end{proof}

}


\section{Construction of pairwise reversible Gibbs-preserving Operations (Proof of Theorem~\ref{thm:sufficient condition reversible})}\label{app:reversible}

We first show a general sufficient condition for the existence of pairwise reversible Gibbs-preserving Operations. 
To this end, let $\psi$ be a pure state in $S$. 
We define the min-relative entropy with respect to the Gibbs state by~\cite{Datta2009min,Regula2018convex,Buscemi2019information,Liu2019one-shot} 
\bal
 D_{\min}(\psi\|\tau_S) = -\log\Tr(\psi\tau_S).
\eal
Also, for an arbitrary state $\rho$ in $S$, max-relative entropy with respect to the Gibbs state is defined by~\cite{Datta2009min,Regula2018convex,Buscemi2019information,Liu2019one-shot} 
\bal
 D_{\max}(\rho\|\tau_S) = \log \min\lset s \sbar \rho\leq s\tau_S\rset.
\eal

We then obtain the following result. 

\begin{theorem}\label{thm:general sufficient app}
Let $S$ be a system with Hamiltonian $H_S=\sum_n E_{n,S}\dm{n}$ and $S'$ be a system with some arbitrary Hamiltonian. 
If there exists a pair of pure states $\psi\in S$ and $\phi\in S'$ such that  
\bal
\bra{i}\psi\ket{i} \neq 0,\quad \bra{j}\psi\ket{j} \neq 0,\quad E_{i,S}\neq E_{j,S}
\label{eq:coherence condition}
\eal
and 
\bal
 D_{\min}(\psi\|\tau_S) = D_{\min}(\phi\|\tau_{S'}) = D_{\max}(\phi\|\tau_{S'}),
 \label{eq:Dmin is Dmax}
\eal
there exists a pairwise reversible Gibbs-preserving Operation from $S$ to $S'$.
\end{theorem}

\begin{proof}
Let $\psi$ and $\phi$ be the pure states in \eqref{eq:coherence condition} and \eqref{eq:Dmin is Dmax}.
Define
\bal
 \Lambda(\rho)=\Tr(\psi\rho)\phi + \Tr\left[(\mathds{1}-\psi)\rho\right]\eta 
 \label{eq:measure and prepare diverging cost}
\eal
where 
\bal
 \eta = \frac{\tau_{S'}-\Tr(\psi\tau_S)\phi}{\Tr\left[(\mathds{1}-\psi)\tau_{S}\right]}.
\eal
The operator $\eta$ is a valid state because it clearly has the unit trace and 
\bal
\phi\leq 2^{D_{\max}(\phi\|\tau_{S'})}\tau_{S'}=2^{D_{\min}(\psi\|\tau_S)}\tau_{S'}=\Tr(\psi\tau_S)^{-1}\tau_{S'}
\eal
where the first inequality is by definition of $D_{\max}$, the first equality is due to \eqref{eq:Dmin is Dmax}, and the last equality is by definition of $D_{\min}$. This ensures that $\eta\geq 0$ and that $\Lambda$ is a valid Gibbs-preserving channel.

Moreover, 
\bal
 \Tr(\phi\eta) \propto \Tr(\phi\tau_{S'})-\Tr(\psi\tau_S) = 2^{-D_{\min}(\phi\|\tau_{S'})} - 2^{-D_{\min}(\psi\|\tau_S)} = 0
\eal
where the final equality is due to \eqref{eq:Dmin is Dmax}.
This implies that $\phi$ and $\eta$ are perfectly distinguishable and thus ensures that $\delta(\Lambda,\bbP)=0$ for a choice of state pair $\bbP=\{\psi,\sigma\}$ where $\sigma$ is an arbitrary state in $S$ such that $\Tr(\psi\sigma)=0$.

As a choice of $\sigma$, we can particularly choose an orthogonal pure state $\psi^\perp$ defined as follows. 
Let $a_i\coloneqq \braket{i}{\psi}$ and $a_j\coloneqq \braket{j}{\psi}$ be the coefficients for $i$\,th and $j$\,th energy of $\psi$, where $i$ and $j$ are the labels in \eqref{eq:coherence condition}.
\HT{We then choose 
\bal
 \ket{\psi^\perp} = \frac{1}{\sqrt{|a_j|^2+|a_i|^2}}\left(a^*_j\ket{i} - a^*_i\ket{j}\right),
 \label{eq:orthogonal pure state}
\eal
for which one can directly check that the condition $\Tr(\psi\psi^\perp)=0$ is satisfied. Here, $a^*_{i,j}$ is the complex conjugate of the coefficients $a_{i,j}$.}
Since
\bal
\Lambda^\dagger(H_{S'}) = \Tr(\phi H_{S'})\psi + \Tr(\eta H_{S'})(\mathds{1}-\psi),
\eal
we get
\bal
 \brakets{\psi}{\Lambda^\dagger(H_{S'})}{\psi^\perp} = 0.
\eal
Therefore,
\bal
 \left|\brakets{\psi}{H-\Lambda^\dagger(H_{S'})}{\psi^\perp}\right|&=\left|\brakets{\psi}{H}{\psi^\perp}\right|= \HT{\frac{|a^*_ia^*_j|}{\sqrt{|a_i|^2+|a_j|^2}}|E_{i,S}-E_{j,S}|}>0
\eal
where the last inequality is because of the assumption that $E_{i,S}\neq E_{j,S}$.
This ensures 
\bal
 \calC(\Lambda,\bbP) = \left|\brakets{\psi}{H-\Lambda^\dagger(H_{S'})}{\psi^\perp}\right|>0
\eal
with a reversible pair $\bbP=\{\psi,\psi^\perp\}$, concluding the proof.
\end{proof}

Theorem~\ref{thm:general sufficient app} admits the following simple sufficient condition. 

\begin{corollary}[Theorem \ref{thm:sufficient condition reversible} in the main text]\label{coro_thm2}
Let $\tau_{X,i}=\brakets{i}{\tau_X}{i}_X$ be the Gibbs distribution for the Gibbs state for a system $X$ with Hamiltonian $H_X = \sum_i E_{X,i}\dm{i}_X$. Then, if there are integers $i$, $j$, and $i'$ for systems $S$ and $S'$ such that
\bal
 \tau_{S,i} < \tau_{S',i'} < \tau_{S,j},
 \label{eq:condition sufficient_SM}
\eal
there exists a pairwise reversible Gibbs-preserving Operation $\Lambda:S\to S'$ and a reversible pair $\bbP$ such that $\calC(\Lambda,\bbP)>0$.
\end{corollary}

\begin{proof}
Let $r\in(0,1)$ be a real positive number satisfying 
\eq{
\tau_{S',i'}=r\tau_{S,i}+(1-r)\tau_{S,j}\label{cond1'},
}
whose existence is guaranteed because of \eqref{eq:condition sufficient_SM}.
Let $\psi$ and $\phi$ be the pure states defined by 
\eq{
\ket{\psi}&:=\sqrt{r}\ket{i}_S+\sqrt{1-r}\ket{j}_S,\\
\ket{\phi}&:=\ket{i'}_{S'}.
}
It suffices to show that $\psi$ and $\phi$ satisfy \eqref{eq:coherence condition} and \eqref{eq:Dmin is Dmax}.
It is straightforward to check \eqref{eq:coherence condition} noting $0<r<1$.
Eq.~\eqref{eq:Dmin is Dmax} can be checked as follows:
\eq{
D_{\min}(\psi\|\tau_S)&=-\log\Tr(\psi\tau_S)\nonumber\\
&=-\log(r\tau_{S,i}+(1-r)\tau_{S,j})\nonumber\\
&=-\log\tau_{S',i'}\nonumber\\
&=-\log\Tr[\phi\tau_{S'}]=D_{\min}(\phi\|\tau_{S'})\nonumber\\
&=\log \min\lset s \sbar \ket{i'}\bra{i'}_{S'}\leq s\tau_{S'}\rset=D_{\max}(\phi\|\tau_{S'}).
}
\end{proof}


We also provide an alternative construction. 

\begin{proposition}
For a countably infinite series $\{E_n\}_n$ of real numbers, let $S_d(\{E_n\}_n)$ be an arbitrary $d$-dimensional system equipped with Hamiltonian $H_d=\sum_{n=0}^{d-1}E_n\dm{n}$. 
Then, for arbitrary $d\geq 3$ and $d'\leq d$, and an arbitrary energy spectrum $\{E_n\}_n$ with $E_i\geq E_j,\ \forall i,j$ such that it is not fully degenerate above the ground energy, i.e., there exists $1\leq i\leq d-1$ such that $E_{i+1}>E_i$, there exists a pairwise reversible Gibbs-preserving operation $\Lambda:S_d(\{E_n\}_n)\to S_{d'}(\{E_n\}_n)$ with a reversible pair $\bbP$ such that $\calC(\Lambda,\bbP)>0$. 
\end{proposition}

\begin{proof}
Let $\tau_{S}=\sum_n \tau_{S,n}$ be the Gibbs state for $S_d(\{E_n\}_n)$ and $\tau_{S'}=\sum_n \tau_{S',n}$ be the Gibbs state for $S_{d'}(\{E_n\}_n)$.
Let $i$ be an integer such that $1\leq i \leq d-1$ and $E_{i+1}>E_i$ ensured by the assumption.
Define
\bal
 \Lambda(\rho) =  \Tr(\dm{+}_{i,i+1} \rho)\dm{0} + \Tr(\dm{-}_{i,i+1} \rho)\dm{1} + \Tr(P_{i,i+1}^\perp\rho)\eta
\eal
where $P_{i,i+1}^\perp = \mathds{1}_d-(\dm{i}+\dm{i+1})$ is the projector onto the input space complement to ${\rm span}\{\ket{i},\ket{i+1}\}$, and
\bal
 \eta \coloneqq \frac{\tau_{S'}-\Tr(\dm{+}_{i,i+1}\tau_S)\dm{0}-\Tr(\dm{-}_{i,i+1}\tau_S)\dm{1}}{\Tr(P_{i,i+1}^\perp \tau_{S})}.
\eal
The operator $\eta$ is a valid state because it clearly has a unit trace and 
\bal
\eta \propto \left(\tau_{S',0}- \frac{\tau_{S,i}+\tau_{S,i+1}}{2}\right)\dm{0} + \left(\tau_{S',1}- \frac{\tau_{S,i}+\tau_{S,i+1}}{2}\right)\dm{1}+ \sum_{j\geq 2}\tau_{S',j} \dm{j}\geq 0
\eal
where we used the fact that for any $k$,
\bal
 \tau_{S',k} - \tau_{S,k} = \frac{(Z-Z')e^{-\beta E_k}}{ZZ'} \geq 0,\quad Z\coloneqq\sum_{n=0}^{d-1} e^{-\beta E_n}, Z'\coloneqq\sum_{n=0}^{d'-1}e^{-\beta E_n}
\eal
because $d\geq d'$ by assumption, which implies $Z\geq Z'$.
We then get 
\bal
 \tau_{S',0}- \frac{\tau_{S,i}+\tau_{S,i+1}}{2} \geq \tau_{S,0}- \frac{\tau_{S,i}+\tau_{S,i+1}}{2}\geq 0\\
 \tau_{S',1}- \frac{\tau_{S,i}+\tau_{S,i+1}}{2} \geq \tau_{S,1}- \frac{\tau_{S,i}+\tau_{S,i+1}}{2}\geq 0
\eal
where the last inequalities hold because $E_{i+1}\geq E_i\geq E_1\geq E_0$ and thus $\tau_{S,i+1}\leq \tau_{S,i} \leq \tau_{S,1}\leq \tau_{S,0}$.
This ensures that $\Lambda$ is a measure-and-prepare channel, and due to the definition of $\eta$, $\Lambda$ is Gibbs-preserving. 

In addition, a state pair $\{\dm{+}_{i,i+1}, \dm{-}_{i,i+1}\}$ is reversible because $\Lambda(\dm{+}_{i,i+1})=\dm{0}$ and $\Lambda(\dm{-}_{i,i+1})=\dm{1}$ are perfectly distinguishable. This ensures $\delta(\Lambda,\bbP) = 0$ for $\bbP=\{\dm{+}_{i,i+1}, \dm{-}_{i,i+1}\}$.

One can also check $\calC(\Lambda,\bbP)>0$ as follows. 
We have
\bal
\Lambda^\dagger(H_{d'}) = \Tr(H_{d'}\dm{0})\dm{+}_{i,i+1} + \Tr(H_{d'}\dm{1})\dm{-}_{i,i+1} + \Tr(H_{d'}\eta)P_{i,i+1}^\perp
\eal
resulting in 
\bal
 {}_{i,i+1}\brakets{+}{\Lambda^\dagger(H_{d'})}{-}_{i,i+1} = 0.
\eal
On the other hand,
\bal
 \left|{}_{i,i+1}\brakets{+}{H_d}{-}_{i,i+1}\right| = \frac{E_{i+1}-E_i}{2}>0
\eal
because $E_{i+1}> E_i$ by assumption. This ensures $\calC(\Lambda,\bbP)>0$. 

\end{proof}


\section{Upper bound for the coherence cost of the channel given in Eq.~(\ref{eq:Faist example})}\label{App:cost_Faist}
We show that the coherence cost $\calF_c(\Lambda)$ for the channel $\Lambda:S\to S'$ defined by
\bal
 \Lambda(\rho) = \brakets{1}{\rho}{1}\eta + \brakets{0}{\rho}{0}\sigma
 \label{eq:Faist example_supple}
\eal
satisfies the upper bound
\eq{
\calF_c(\Lambda)\le\calF(\eta)+\calF(\sigma).\label{upp_Faist}
}

\begin{proof}
Let $E$ be an ancillary system with $H_E = 0$. 
Let $\Lambda_1:S\to E$ and $\Lambda_2:E\to S'$ be the channels defined by 
\eq{
\Lambda_1(\kappa_S):=\bra{1}\kappa_S\ket{1}\dm{1}_E+\bra{0}\kappa_S\ket{0}\dm{0}_E
}
and 
\eq{
\Lambda_2(\kappa_E):=\bra{1}\kappa_E\ket{1}\eta+\bra{0}\kappa_E\ket{0}\sigma
}
for arbitrary states $\kappa_S$ in $S$ and $\kappa_E$ in $E$. 

Let $\tilde S$ be a system identical to $S$ (equipped with Hamiltonian $H_S$).
The channels $\Lambda_1$ and $\Lambda_2$ can be implemented by unitaries $U$ on $SE$ and $V$ on $ES\tilde S$ by 
\eq{
\Lambda_1(\kappa_S)&=\Tr_{S}[U\kappa_S\otimes\ket{0}\bra{0}_{E}U^\dagger]\\
\Lambda_2(\kappa_E)&=\Tr_{E\tilde S}[V\kappa_E\otimes\eta\otimes\sigma V^\dagger]
}
where 
\eq{
U&:=\ket{00}\bra{00}_{SE}+\ket{11}\bra{10}_{SE}+\ket{01}\bra{01}_{SE}+\ket{10}\bra{11}_{SE},\\
V&:=\ket{0}\bra{0}_E\otimes\mathds{1}_{S\tilde S}+\ket{1}\bra{1}_E\otimes U_{\mathrm{SWAP}}
}
where $U_{\mathrm{SWAP}}$ is the swap operator between $S$ and $\tilde S$. 
Because of $H_E=0$ and $H_S=H_{\tilde S}$, the relations $[U,H_S+H_E]=0$ and $[V,H_{E}+H_{S}+H_{\tilde S}]=0$ are satisfied.
Therefore, we obtain
\eq{
\calF_c(\Lambda_1)&\leq \calF(\dm{0})=0\\
\calF_c(\Lambda_2)&\le\calF(\eta\otimes\sigma)=\calF(\eta)+\calF(\sigma).
}
Since $\calF_c(\Lambda)\le\calF_c(\Lambda_1)+\calF_{c}(\Lambda_2)$, we obtain \eqref{upp_Faist}.
\end{proof}

\HT{We note that \eqref{upp_Faist} is only an upper bound of the cost of the channel in (\ref{eq:Faist example}) of the main text, which may not be tight in general.
}


\section{On atypicality of cost-diverging Gibbs-preserving Operations}\label{app:atypicality}

Theorem~\ref{thm:lower bound app} and Corollary~\ref{coro_thm2} show that there are infinitely many cost-diverging Gibbs-preserving Operations. However, this does not imply that the ones with diverging coherence cost are typical instances among all Gibbs-preserving Operations. 
Here, we show that, in fact, cost-diverging operations are measure zero with respect to the set of all Gibbs-preserving Operations.
Nevertheless, we stress that one still needs to be cautious about the implementability of Gibbs-preserving Operations because (1) in many cases, we are interested in highly structured Gibbs-preserving Operations, which can be cost-diverging ones, and (2) even if a desired Gibbs-preserving Operation has a finite coherence cost, it can still come with a large coherence cost. Indeed, Theorem~\ref{thm:lower bound app} implies that the operations that approximate a cost-diverging one with accuracy $\epsilon$ must come with a coherence cost proportional to $1/\epsilon$, which forms a set with nonzero measure.

We show the atypicality of cost-diverging instances by proving that they only reside on the boundary of the set of Gibbs-preserving Operations. 
To formalize this, we need the notion of the interior points of a subset of vector spaces, where the subset of interest itself could also have zero measure. 
This is represented by the relative interior of the subset~\cite{rockafellar2015convex}. 
For a vector space equipped with trace norm, define the open ball centered around a vector $v$ by
\bal
\calB^\epsilon(v)\coloneqq \lset u  \sbar \|u-v\|_1<\epsilon\rset.
\eal
Also, we define the intersection between the open ball and the affine hull of a set $\calS$ by 
\bal
\calB^\epsilon_\calS(v) \coloneqq \calB^\epsilon(v)\cap {\rm aff}(\calS) 
\eal
where ${\rm aff}(\calS)=\lset s \sbar s=(1-t)s_1 + ts_2,\ s_1\in \calS, s_2\in\calS,t\in\mathbb{R}\rset$.
Then, the relative interior of the set $\calS$ is defined by 
    \bal
     {\rm relint}(\calS)\coloneqq \lset v\in\cal S \sbar \exists \epsilon >0\mbox{ s.t. }\calB_\calS^\epsilon(v)\subseteq \calS \rset.
    \eal

In the following, we show that every operation in the relative interior of the set of Gibbs-preserving Operations has a finite coherence cost, which implies that the cost-diverging ones are measure zero with respect to all Gibbs-preserving Operations because any finite-dimensional convex set has a nonzero relative interior~\cite[Theorem 6.2]{rockafellar2015convex}. 
To this end, we begin by the following result, which we will employ to show that if all members in the set can be approximated arbitrarily well, all interior points of this set can be exactly implemented.
This was shown in Ref.~\cite[Corollary 7]{Wilming2022correlationsin} for interior points of the full-dimensional subset, and here we extend it to relative interior points. 
Nevertheless, the proof is almost identical, and here we particularly follow the descriptions employed in Ref.~\cite[Lemma S.14]{Shiraishi2024arbitrary}, which explained the proof in a bit more detail.

\begin{lemma}\label{lem:approximate to exact}
Let $\mathcal{V}$ be a subset of linear operators on a finite-dimensional Hilbert space, and let $\{\mathcal{S}_n\}_{n=1}^\infty$ be a family of closed convex sets satisfying $\mathcal{S}_n\subseteq \mathcal{V}$ and $\mathcal{S}_n\subseteq \mathcal{S}_{n+1}$ for every $n$. 
Also, suppose that $\{\mathcal{S}_n\}_{n=1}^\infty$ approximates $\mathcal{V}$ arbitrarily well, in the sense that for an arbitrary $\epsilon>0$, there exists a sufficiently large $N$ such that $\min_{u\in \mathcal{S}_N}\|v-u\|_1<\epsilon$ for every $v\in \mathcal{V}$. Then, for every $\kappa\in {\rm relint}(\mathcal{V})$, there exists an integer $n$ such that $\kappa\in \mathcal{S}_n$.
\end{lemma}
\begin{proof}

Since $\kappa\in{\rm relint}(\mathcal{V})$, there exists a sufficiently small $\delta>0$ such that $B_\mathcal{V}^\delta(\kappa)\subseteq \mathcal{V}$.
By assumption of $\{\mathcal{S}_n\}_n$, there exists an integer $N_{\delta/2}$ such that $\min_{u\in \mathcal{S}_{N_{\delta/2}}}\|v-u\|_1<\delta/2$ for an arbitrary $v\in \mathcal{V}$.

In the following, we aim to show that $\kappa\in \mathcal{S}_{N_{\delta/2}}$.
Suppose contrarily that $\kappa\not\in \mathcal{S}_{N_{\delta/2}}$.
Let $s\in \mathcal{S}_{N_{\delta/2}}$ be the closest point in $\mathcal{S}_{N_{\delta/2}}$ from $\kappa$, i.e., 
\bal
s={\rm argmin}_{s'\in \mathcal{S}_{N_{\delta/2}}} \|\kappa-s'\|_1.
\label{eq:closest point}
\eal

Letting $r(t)=(1-t)s+t\kappa$, 
it holds that for every $\delta'<\delta$ there exists $t'\geq 1$ such that $\|r(t')-\kappa\|_1=(t'-1)\|s-\kappa\|_1= \delta'$. 
Together with the fact that $r(t)\in{\rm aff}(\mathcal{V})$, which is because $s\in \mathcal{S}_{N_{\delta/2}} \subseteq \mathcal{V}$ and $\kappa\in \mathcal{V}$, we have $r(t')\in\calB^\delta_\mathcal{V}(\kappa)\subseteq \mathcal{V}$.
Choosing $\delta'=2\delta/3$ in particular, we get 
\bal
 \|r(t')-s\|_1=t'\|s-\kappa\|_1\geq (t'-1)\|s-\kappa\|_1=\|r(t')-\kappa\|_1=2\delta/3.
 \label{eq:ray and the point in the set}
\eal
On the other hand,
\bal
 \min_{u\in \mathcal{S}_{N_{\delta/2}}}\|r(t')-u\|_1&=\min_{u\in \mathcal{S}_{N_{\delta/2}}} \|-(t'-1)s+t'\kappa-u\|_1\\
 &= t'\min_{u\in \mathcal{S}_{N_{\delta/2}}}\left\|\kappa-\left[\left(1-\frac{1}{t'}\right)s+\frac{1}{t'}u\right]\right\|_1 = t'\|\kappa-s\|_1
 \label{eq:distance ray and set}
\eal
where in the last equality, we used the fact that $(1-\frac{1}{t'})s+\frac{1}{t'}u\in \mathcal{S}_{N_{\delta/2}}$ because of the convexity of $\mathcal{S}_{N_{\delta/2}}$ and \eqref{eq:closest point}, ensuring $u=s$ achieves the minimum. 
Combing \eqref{eq:ray and the point in the set} and \eqref{eq:distance ray and set}, we get 
\bal
\min_{u\in \mathcal{S}_{N_{\delta/2}}}\|r(t')-u\|_1\geq 2\delta/3.
\label{eq:distance ray and set summary}
\eal
Since the definition of $\mathcal{S}_{N_{\delta/2}}$ implies $\delta/2\geq \min_{u\in \mathcal{S}_{N_{\delta/2}}}\|r(t')-u\|_1$, we reach a contradiction $\delta/2\geq 2\delta/3$ together with \eqref{eq:distance ray and set summary}.
\end{proof}

The next result shows that if a subset contains an interior point of its superset, no interior points of the subset cannot be on the boundary of the superset.
We will employ this to ensure that relative interior of the set of Gibbs-preserving Operations are also in the interior of all quantum channels. 

\begin{lemma} \label{lem:relative interior}
    Let $\calS\subseteq \calV$ be a subset that contains at least one interior point $v\in{\rm relint}(\calV)$ with respect to $\calV$. Then, every $s\in{\rm relint}(\calS)$ is also a relative interior point of $\calV$, i.e., $s\in{\rm relint}(\calV)$.
\end{lemma}
\begin{proof}
    For an arbitrary $s\in{\rm relint}(\calS)$, there exists $\epsilon>0$ such that $\calB_\calV^\epsilon(s)\subseteq \calV$. For $v\in{\rm relint}(\calV)$, let $r(t) = -tv + (1+t)s$ for $t\geq 0$. Since $r(t)\in {\rm aff}(\calS),\forall t$ and $s\in{\rm relint}(\calS)$, there exists sufficiently small $\tilde t>0$ such that $r(\tilde t)\in\calS\subseteq \calV$. 
    Recalling the line segment property~\cite[Theorem 6.1]{rockafellar2015convex}, which states that for an arbitrary closed convex set $\calC$ and arbitrary points $x\in{\rm relint}(\calC)$ and $y\in\calC$, all points $\lambda x + (1-\lambda)y$ with $0<\lambda\leq 1$ are in the relative interior of $\calC$, $s=\frac{r(\tilde t)+\tilde tv}{1+\tilde t}$ satisfies $s\in{\rm relint}(\calV)$ as $v\in{\rm relint}(\cal V)$ and $r(\tilde t)\in \calV$, as well as $\tilde t>0$. 
\end{proof}

We are in the position to show the desired statement by combining Lemmas~\ref{lem:approximate to exact} and \ref{lem:relative interior}, together with the fact that every quantum channel can be approximated arbitrarily well with a finite coherence cost ensured by Theorem~\ref{thm:general upper app}. 

\begin{proposition}
    Let $\mathbb{O}_{\rm GPO}$ be the set of Gibbs-preserving Operations. Then, every coherence-diverging Gibbs-preserving Operation $\Lambda$ must be on the boundary of $\mathbb{O}_{\rm GPO}$, i.e., $\Lambda\not\in{\rm relint}(\mathbb{O}_{\rm GPO})$.
\end{proposition}
\begin{proof}
We show that every $\Lambda\in{\rm relint}(\mathbb{O}_{\rm GPO})$ can be implemented exactly with a finite coherence cost. 
For a subset $\mathbb{O}$ of channels, let $\tilde{\mathbb{O}}$ denote the set of Choi operators for $\mathbb{O}$.
Let $\mathbb{O}_{\rm GPO}(S\to S')$ be the set of Gibbs-preserving Opeations from $S$ to $S'$, and let $\mathbb{T}(S\to S')$ be the set of all channels from $S$ to $S'$.

We first note that $\mathbb{O}_{\rm GPO}(S\to S')$ contains a channel $\Xi$ such that $J_{\Xi}\in{\rm relint}(\tilde{\mathbb{T}}(S\to S'))$ where $J_\Xi$ is the Choi operator of $\Xi$.
For instance, take $\Xi(\cdot)=\Tr(\cdot)\tau_{S'}$, whose Choi operator is $J_\Xi=\id\otimes\Xi(\tilde\Phi)=\mathbb{I}\otimes \tau_{S'}$ where $\tilde\Phi=\sum_{i,j=0}^{d_S-1}\ketbra{ii}{jj}$ is the unnormalized maximally entangled state.
Then, any operator in ${\rm aff}(\tilde{\mathbb{T}}(S\to S'))=\lset J_{SS'} \sbar \Tr_{S'}(J)=\mathbb{I}_S\rset$ in the neighborhood of $J_\Xi$ is also a valid Choi operator because $J_\Xi$ is full rank and thus $J_\Xi>0$.
Threrefore, we can apply Lemma~\ref{lem:relative interior} to ensure that for every $\Lambda\in{\rm relint}(\mathbb{O}_{\rm GPO}(S\to S'))$, which is equivalent to $J_\Lambda\in{\rm relint}(\tilde{\mathbb{O}}_{\rm GPO}(S\to S'))$, it holds that $J_\Lambda\in{\rm relint}(\tilde{\mathbb{T}}(S\to S'))$.

 Let $\mathcal{S}_n$ be the set of Choi operators of channels with input and output systems $S$ and $S'$ whose coherence cost is upper bounded by the integer $n$, i.e.,
  \bal
   \mathcal{S}_n \coloneqq {\rm cl}\lset J_{\Xi}\in\tilde{\mathbb{T}}(S\to S')\sbar  \calF_c^{\epsilon=0}(\Xi)\leq n\rset
  \eal
  where ${\rm cl}$ refers to the closure of the set, whose necessity comes from that in the implementation $\Xi(\cdot)=\Lambda(\cdot\otimes \eta),\Lambda\in\mathbb{O}_{\rm TO}$, the choice of the system that $\eta$ acts on is arbitrary and thus the set of states $\eta$ that admits $\calF_c^{\epsilon=0}(\Xi)\leq n$ may not be closed.

  We can see that $\mathcal{S}_n$ is convex, following a similar argument to the one in \cite[Appendix C]{Lostaglio2015quantum}. Consider $\Xi_1, \Xi_2\in\calS_n$ that admit implementation $\Xi_i(\cdot)=\Tr_{\overline S'}[U_i(\cdot\otimes \eta_i\otimes \tau_i)U_i^\dagger],i=1,2$ where $U_i$ is an energy-conserving unitary on $SR_iE_i$, $\tau_i$ is the thermal state of some system $E_i$, and $\eta_i$ is a coherent resource state on some system $R_i$ with $\calF(\eta_i)\leq n$. 
  Note that $\Tr_{\overline S'}$ denotes the partial trace over systems other than $S'$.
  It then suffices to show that $\Xi=p\Xi_1+(1-p)\Xi_2$ can be implemented by a Thermal Operation aided by a resource state $\eta$ with $\calF(\eta)\leq n$.
  To this end, let $\eta=p\dm{0}_F\otimes\tilde\eta_1+(1-p)\dm{1}_F\otimes \tilde\eta_2$ where $\tilde\eta_1=\eta_1\oplus 0_{R_2}$ and $\tilde\eta_2=0_{R_1}\oplus \eta_2$ are states on $R=R_1\oplus R_2$ with Hamiltonian $H_R=H_{R_1}\oplus H_{R_2}$, and $F$ is the system introduced for a flag with trivial Hamitlonian $H_F=0$. 
  Let $U=\dm{0}_F\otimes \tilde U_1\otimes \mathbb{I}_{E_2} + \dm{1}_F\otimes \tilde U_2\otimes \mathbb{I}_{E_1}$ where $\tilde U_1=U_1\oplus \mathbb{I}_{SR_2E_1}$ and $\tilde U_2=U_2\oplus \mathbb{I}_{SR_1E_2}$.
  This is energy-conserving because 
  \begin{equation}\begin{aligned}
  &[U,H_F+H_S+H_{R}+H_{E_1}+H_{E_2}] \\
  &= \dm{0}\otimes [U_1,H_S+H_{R_1}+H_{E_1}]+\dm{1}\otimes[U_2,H_S+H_{R_2}+H_{E_2}]
  = 0.
  \end{aligned}\end{equation}
Then,
\bal
 \Tr_{\overline S'}[U(\rho\otimes\eta\otimes\tau_1\otimes\tau_2) U^\dagger]&=p\Tr_{\overline S'}[U_1(\rho\otimes\eta_1\otimes\tau_1) U_1^\dagger]+(1-p)\Tr_{\overline S'}[U_2\rho\otimes\eta_2\otimes\tau_2 U_2^\dagger]\\
 & = p\Xi_1 + (1-p)\Xi_2.
\eal  
Noting that $\calF(\eta)=p\calF(\eta_1)+(1-p)\calF(\eta_2)\leq n$ concludes the proof that $\calS_n$ is convex.

Theorem~\ref{thm:general upper app} reveals that every quantum channel can be implemented with an arbitrary accuracy $\epsilon$ with a finite coherence cost specified in the right-hand side of \eqref{eq:coherence cost upper bound general}.
This particularly means that every Choi operator in $\tilde{\mathbb{T}}(S\to S')$ can be approximated with arbitrary accuracy by a Choi operator in $\calS_N\subseteq \tilde{\mathbb{T}}(S\to S')$ for a sufficiently large $N$.  
Therefore, applying Lemma~\ref{lem:approximate to exact} with $\mathcal{V}=\tilde{\mathbb{T}}(S\to S')$, we get that for every $J_\Lambda\in{\rm relint}(\tilde{\mathbb{O}}_{\rm GPO}(S\to S'))\subseteq{\rm relint}(\tilde{\mathbb{T}}(S\to S'))$, where the second inclusion is guaranteed by Lemma~\ref{lem:relative interior}, and for every $\delta>0$ such that $\calB_\calV^\delta(J_\Lambda)\subseteq {\rm relint}(\calV)$, there exists an integer $n$ such that $\calB_\calV^\delta(J_\Lambda)\subseteq\calS_n$. 
This particularly means that $J_\Lambda\in{\rm relint}(\calS_n)$ and thus $\calF_c^{\epsilon=0}(\Lambda)\leq n$, i.e., $\Lambda$ can be implemented exactly with a finite coherence cost. 

\end{proof}


\section{State transformations requiring cost-diverging Gibbs preserving Operations (Proof of Theorem~\ref{thm:restrictions on state transition})}\label{app:state transition}
Besides the physical implementability of a given Gibbs-preserving Operation, another significant problem is whether a given state transformation $\rho\to\sigma$ that can be realized by a Gibbs-preserving Operation can physically be achieved by a finite coherence cost. 
Theorem~\ref{thm:lower bound}~or~\ref{thm:sufficient condition reversible} does not directly address this question because there are typically many quantum channels that can realize the said state transformation.
To investigate the smallest coherence cost for a given state transformation, one needs to optimize over all such quantum channels, which is generally intractable. 

Moreover, it can easily be seen that an arbitrary state transformation $\rho\to\sigma$ \emph{can} always be realized by a finite coherence cost because if one just prepares the state $\sigma$ in the auxiliary system equipped with the same Hamiltonian as the one for the main system and swaps $\rho$ and $\sigma$, this realizes the desired state transformation with the coherence cost $\calF(\sigma)$.
Therefore, we need to impose additional structure to obtain nontrivial insights. 

To this end, we focus on the setting where an operation implemented by a Thermal Operation with coherence assistance is also Gibbs-preserving.  
This is a reasonable restriction given that we are considering a state transformation that can be realized by a Gibbs-preserving Operation.
Here, we employ our methods to show that some state transformations realized by Gibbs-preserving Operations can never be achieved by a Thermal Operation and finite coherence, under the restriction that the implemented opearation itself is Gibbs-preserving. 
\HT{We also show that the cost-diverging state transitions given by our method are costly even when we allow small errors in their realization.
 To describe the result, we define the cost of approximate state transitions as follows:
 \eq{
\calF^{\epsilon}(\rho\rightarrow\sigma):=\min\lset\calF_c(\Lambda)\sbar\Lambda\in\mathbb{O}_{\rm GPO},\ \sigma'=\Lambda(\rho'),\ D_F(\rho,\rho')\le\epsilon,\ D_F(\sigma,\sigma')\le\epsilon\rset
 }
Here, $\mathbb{O}_{\rm GPO}$ is the set of Gibbs-preserving Operations.}

\begin{theorem}[\HT{A generalization of} Theorem \ref{thm:restrictions on state transition} in the main text]\label{thm:restrictions on state transition_SM}
Let $\tau_{X,i}=\brakets{i}{\tau_X}{i}_X$ be the Gibbs distribution for the Gibbs state for a system $X$ with Hamiltonian $H_X = \sum_i E_{X,i}\dm{i}_X$. Then, if there are integers $i$, $j$, and $i'$ for systems $S$ and $S'$ such that
\bal
 \tau_{S,i} < \tau_{S',i'} < \tau_{S,j},
 \label{eq:condition sufficient_for_ST}
\eal
\HT{there is a pair $(\rho,\sigma)$ of states on $S$ and $S'$ satisfying (i)$\rho\rightarrow\sigma$ is possible by a Gibbs-preserving Operation from $S$ to $S'$, (ii) $\rho\rightarrow\sigma$ is impossible by any finite-cost Gibbs-preserving Operations from $S$ to $S'$, i.e. $\calF^{\epsilon=0}(\rho\rightarrow\sigma)=\infty$, and (iii) for non-zero error $\epsilon$, $\rho\rightarrow\sigma$ satisfies
\eq{
\calF^{\epsilon}(\rho\rightarrow\sigma)\ge\frac{\kappa(\epsilon,\beta,\tau_{S,i},\tau_{S',i'},\tau_{S,j})}{\epsilon}-(\Delta(H_{S})+3\Delta(H_{S'}))^2,\label{eq:cost_approx_trans}
}
where $\kappa(\beta,\tau_{S,i},\tau_{S',i'},\tau_{S,j})$ is the following real-valued function that is positive when \eqref{eq:condition sufficient_for_ST} holds:
 \eq{ \kappa(\beta,\tau_{S,i},\tau_{S',i'},\tau_{S,j}):=\frac{(\tau_{S,i}+\tau_{S,j}-\tau_{S',i'})(\tau_{S,j}-\tau_{S',i'})(\tau_{S',i'}-\tau_{S,i})\log\frac{\tau_{S,j}}{\tau_{S,i}}}{2\beta(\tau_{S,j}-\tau_{S,i})^2(\tau_{S,i}+\tau_{S,j}+2\sqrt{(\tau_{S,j}-\tau_{S',i'})(\tau_{S',i'}-\tau_{S,i})})}
 }
 }
\end{theorem}
\HT{We remark that (i) and (ii) correspond to the Theorem \ref{thm:restrictions on state transition} in the main text.}

\begin{proof}
Consider the following state transition
\eq{
\ket{\eta_+}_S\bra{\eta_+}_S\rightarrow\ket{i'}_{S'}\bra{i'}_{S'}\label{eq:state transformation infinite cost}
}
where 
\eq{
\ket{\eta_+}_S&:=\sqrt{r}\ket{i}_{S}+\sqrt{1-r}\ket{j}_{S},\\
r&:=\frac{\tau_{S',i'}-\tau_{S,j}}{\tau_{S,i}-\tau_{S,j}}.\label{condr}
}

We first show that there exists a Gibbs-preserving Operation that realizes the transformation \eqref{eq:state transformation infinite cost}.
Consider a map $\Lambda_0:S\to S'$ defined by 
\eq{
\Lambda_0(...)&:=\Tr[\kb{\eta_+}...]\kb{i'}_{S'}+\Tr[(1-\kb{\eta_+})...]\xi_{S'},\\
\xi_{S'}&:=\frac{\tau_{S'}-\Tr[\kb{\eta_+}\tau_{S}]\kb{i'}_{S'}}{1-\Tr[\kb{\eta_+}\tau_{S}]}.
}
$\Lambda_0$ is a valid quantum channel because the fact that $\tau_{S'}\geq \left(\tau_{S',i}-\Tr[\kb{\eta_+}\tau_S]\right)\dm{i}=0$ ensures $\xi_{S'}\geq 0$. 
One can also directly see that $\Lambda_0$ is Gibbs-preserving satisfying $\Lambda_0(\tau_S)=\tau_{S'}$. 

\HT{Next, we show (ii), i.e., the state transition cannot be achieved by finite-cost Gibbs-preserving operations.} It suffices to show that an arbitrary Gibbs-preserving Operation $\Lambda:S\to S'$ that realizes the transformation \eqref{eq:state transformation infinite cost} must come with diverging coherence cost, i.e., $\calF_c^{\epsilon=0}(\Lambda)=\infty$. 
Let $\{M_k\}_k$ be a Kraus representation of $\Lambda$.
We write each Kraus operator with respect to energy eigenbases $\{\ket{l}_S\}_{l=1}^{d_S}$ of $S$ and $\{\ket{m}_{S'}\}_{m=1}^{d_{S'}}$ of $S'$ as
\eq{
M_k=&a_k\ket{i'}_{S'}\bra{\eta_+}_S
+\sum_{m\neq i'}b_{k,m}\ket{m}_{S'}\bra{\eta_+}_S
+c_k\ket{i'}_{S'}\bra{\eta_-}_{S}\nonumber\\
&+\sum_{m\neq i'}d_{k,m}\ket{m}_{S'}\bra{\eta_-}_S
+\sum_{l\neq i,j}e_{k,l}\ket{i'}_{S'}\bra{l}_S
+\sum_{m\neq i'}\sum_{l\neq i,j}f_{k,m,l}\ket{m}_{S'}\bra{l}_S,
}
where $\ket{\eta_-}:=\sqrt{1-r}\ket{i}_S-\sqrt{r}\ket{j}_S$. 

The assumption $\Lambda(\dm{\eta_+})=\dm{i'}$ imposes
\eq{
b_{k,m}=0,\enskip\forall k, \enskip\forall m\mbox{ s.t. }m\neq i'.\label{acsum-2}
}
Furthermore, $\sum_kM^\dagger_kM_k=\mathds{1}_S$ gives
\eq{
\sum_k|a_k|^2&=\sum_k\bra{\eta_+}M^\dagger_kM_k\ket{\eta_+}=1,\label{asum1}\\
\sum_kc^*_ka_k&=\sum_k\bra{\eta_-}M^\dagger_kM_k\ket{\eta_+}=0,\label{acsum0}\\
\sum_ke^*_{k,l}a_k&=\sum_k\bra{l}M^\dagger_kM_k\ket{\eta_+}=0,\enskip\forall l\mbox{ s.t. }l\neq i,j\label{acsum-1}\\
\sum_ke^*_{k,l}c_k&=\sum_k\bra{l}M^\dagger_kM_k\ket{\eta_-}=0,\enskip\forall l\mbox{ s.t. }l\neq i,j\label{acsum-3}.
}
Therefore, we obtain
\eq{
\tau_{S',i'}=&\bra{i'}_{S'}\tau_{S'}\ket{i'}_{S'}\nonumber\\
=&\sum_k\bra{i'}_{S'}M_k\tau_S M^\dagger_k\ket{i'}_{S'}\nonumber\\
=&\sum_k|a_k|^2\bra{\eta_+}\tau_S\ket{\eta_+}_S
+\sum_k|c_k|^2\bra{\eta_-}\tau_S\ket{\eta_-}_S
+\sum_{k}\sum_{l\neq i,j}|e_{k,l}|^2\bra{l}_{S}\tau_S\ket{l}_{S}\nonumber\\
&+2\mathrm{Re}\left(\sum_kc^*_ka_k\bra{\eta_+}\tau_S\ket{\eta_-}\right)
+2\mathrm{Re}\left(\sum_{l\neq i,j}\sum_ke^*_{k,l}a_k\bra{\eta_+}\tau_S\ket{l}\right)
+2\mathrm{Re}\left(\sum_{l\neq i,j}\sum_ke^*_{k,l}c_k\bra{\eta_-}\tau_S\ket{l}\right)
\nonumber\\
=&r\tau_{S,i}+(1-r)\tau_{S,j}
+\sum_k|c_k|^2((1-r)\tau_{S,i}+r\tau_{S,j})
+\sum_{k}\sum_{l\neq i,j}|e_{k,l}|^2\tau_{S,l}
\nonumber\\
=&\tau_{S',i'}+\sum_k|c_k|^2((1-r)\tau_{S,i}+r\tau_{S,j})
+\sum_{k}\sum_{l\neq i,j}|e_{k,l}|^2\tau_{S,l}.
}
Because of $(1-r)\tau_{S,i}+r\tau_{S,j}>0$ and $\tau_{S,l}>0$, we get
\eq{
c_k&=0,\enskip\forall k\\
e_{k,l}&=0,\enskip\forall k,\enskip\forall l\mbox{ s.t. }l\neq i,j.
}
Therefore, we can rewrite $M_k$ as
\eq{
M_k=&a_k\ket{i'}_{S'}\bra{\eta_+}_S
+\sum_{m\neq i'}d_{k,m}\ket{m}_{S'}\bra{\eta_-}_S
+\sum_{m\neq i'}\sum_{l\neq i,j}f_{k,m,l}\ket{m}_{S'}\bra{l}_S.\label{diag-pre}
}
This particularly implies that
\eq{
\bra{i'}_{S'}\Lambda(\psi_S)\ket{i'}_{S'}=0,\enskip\forall\ket{\psi}_S\mbox{ s.t. }\braket{\eta_+}{\psi}_S=0.\label{diag}
}

We employ this to construct a reversible input state pair that shows diverging coherence cost via Theorem~\ref{thm:lower bound}.
Let $\bbP=\{\dm{\eta_+},\dm{\eta_-}\}$. 
Then, one can see that $\bbP$ is a reversible pair by considering a recovery channel
\eq{
\calR(...)=\Tr[\kb{i'_{S'}}...]\kb{\eta_{+}}_S+\Tr[(1-\kb{i'_{S'}})...]\kb{\eta_{-}}_S,
}
which perfectly recovers $\ket{\eta_\pm}$ because of \eqref{diag}. 
We can also directly check that $\calC(\Lambda,\mathbb{P})>0$ as follows. 
Note that the operator $\Lambda^\dagger(H_{S'})$ satisfies
\eq{
\bra{\eta_+}_{S}\Lambda^\dagger(H_{S'})\ket{\eta_{-}}_S
&=\sum_{k}\bra{\eta_+}_{S}M^\dagger_kH_{S'}M_k\ket{\eta_{-}}_S\nonumber\\
&=\sum_{m\neq i'}\sum_{k}a^*_{k}E_{S',m}d_{k,m}\delta_{i'm}\nonumber\\
&=0.
}
Therefore,
\eq{
\calC(\Lambda,\mathbb{P})^2&=|\bra{\eta_+}H_{S}\ket{\eta_-}_S|^2\nonumber\\
&=|\bra{\eta_+}(E_{S,i}\kb{i}_S+E_{S,j}\kb{j}_S)\ket{\eta_-}|^2\nonumber\\
&=r(1-r)\left(E_{S,i}-E_{S,j}\right)^2\nonumber\\
&>0
}
where the last inequality is because of the assumption $\tau_{S,i}<\tau_{S',i'}<\tau_{S,j}$, which ensures $E_{S,i}-E_{S,j}\neq 0$ and $0<r<1$.  
Applying Theorem~\ref{thm:lower bound} to $\Lambda$ and $\bbP$ results in $\calF_c^{\epsilon=0}(\Lambda)=\infty$.

\HT{
Next, we show (iii), i.e. $\ket{\eta_+}_S\bra{\eta_+}_S\rightarrow\ket{i'}_{S'}\bra{i'}_{S'}$ satisfies \eqref{eq:cost_approx_trans}. It suffices to show that for an arbitrary state $\sigma'_{S',\epsilon}$ such that $D_F(\ket{i'}_{S'}\bra{i'}_{S'},\sigma'_{S',\epsilon})\le\epsilon$, the following inequality holds:
\begin{align}
\sqrt{\calF^{\epsilon=0}(\ket{\eta_+}_S\bra{\eta_+}_S\rightarrow\sigma'_{S',\epsilon})}\ge\frac{\sqrt{4\kappa(\beta,\tau_{S,i},\tau_{S',i'},\tau_{S,j})}}{\sqrt{\epsilon}}-\Delta(H_{S})-3\Delta(H_{S'}).\label{eq:cost_approx_trans_pre}
 \end{align}
Derivation of \eqref{eq:cost_approx_trans} from \eqref{eq:cost_approx_trans_pre} is as follows: First, due to $(a-b)^2\ge a^2/2-b^2$, we obtain 
 \eq{
 \calF^{\epsilon=0}(\ket{\eta_+}_S\bra{\eta_+}_S\rightarrow\sigma'_{S',\epsilon})\ge\frac{2\kappa(\beta,\tau_{S,i},\tau_{S',i'},\tau_{S,j})}{\epsilon}-(\Delta(H_{S})+3\Delta(H_{S'}))^2.
 }
 Second, if a channel $\tilde{\Lambda}$ satisfies $\tilde{\Lambda}(\rho')=\sigma'_{S',\epsilon}$ for a state $\rho'$ satisfying $D_F(\rho',\ket{\eta_+}_S\bra{\eta_+}_S)\le\epsilon$, the channel $\tilde{\Lambda}$ also satisfies $D_F(\tilde{\Lambda}(\ket{\eta_+}_S\bra{\eta_+}_S),\ket{i'}_{S'}\bra{i'}_{S'})\le2\epsilon$.
 Therefore, $\calF^{\epsilon}(\ket{\eta_+}_S\bra{\eta_+}_S\rightarrow\ket{i'}_{S'}\bra{i'}_{S'})\ge\min_{\sigma''_{S'}:D_F(\sigma''_{S'},\ket{i'}\bra{i'}_{S'})\le2\epsilon}\calF^{0}_{c}(\ket{\eta_+}_S\bra{\eta_+}_S\rightarrow\sigma''_{S'})$. Therefore,  \eqref{eq:cost_approx_trans} follows \eqref{eq:cost_approx_trans_pre}.

 Let us show \eqref{eq:cost_approx_trans_pre}.
 Let $\Lambda$ be an arbitrary Gibbs-preserving operation satisfying
 \eq{ D_F(\Lambda(\ket{\eta_+}\bra{\eta_+}_S),\ket{i'}\bra{i'}_{S'})\le\epsilon.\label{eq:assum_approx}
 }
Since $\calF^{\epsilon=0}(\ket{\eta_+}_S\bra{\eta_+}_S\rightarrow\sigma'_{S',\epsilon})\ge\min\lset\calF_{c}(\Lambda)\sbar D_F(\Lambda(\ket{\eta_+}\bra{\eta_+}_S),\ket{i'}\bra{i'}_{S'})\leq \epsilon\rset$, the inequality \eqref{eq:cost_approx_trans_pre} follows from
\eq{
\sqrt{\calF_{c}(\Lambda)}\ge\frac{\sqrt{4\kappa(\beta,\tau_{S,i},\tau_{S',i'},\tau_{S,j})}}{\sqrt{\epsilon}}-\Delta(H_{S})-3\Delta(H_{S'}).\label{eq:cost_approx_trans_pre2}
}
To show this, we follow a similar strategy of proof of (ii). Let $\{M_k\}_k$ be a Kraus representation of $\Lambda$.
 We write each Kraus operator with respect to energy eigenbases $\{\ket{l}_S\}_{l=1}^{d_S}$ of $S$ and $\{\ket{m}_{S'}\}_{m=1}^{d_{S'}}$ of $S'$ as
 \eq{
 M_k=&a_k\ket{i'}_{S'}\bra{\eta_+}_S
 +\sum_{m\neq i'}b_{k,m}\ket{m}_{S'}\bra{\eta_+}_S
+c_k\ket{i'}_{S'}\bra{\eta_-}_{S}\nonumber\\
 &+\sum_{m\neq i'}d_{k,m}\ket{m}_{S'}\bra{\eta_-}_S
 +\sum_{l\neq i,j}e_{k,l}\ket{i'}_{S'}\bra{l}_S
 +\sum_{m\neq i'}\sum_{l\neq i,j}f_{k,m,l}\ket{m}_{S'}\bra{l}_S,
 }
 where $\ket{\eta_-}:=\sqrt{1-r}\ket{i}_S-\sqrt{r}\ket{j}_S$. 

 The assumption \eqref{eq:assum_approx} imposes $\bra{i'}_{S'}\Lambda(\ket{\eta_+}\bra{\eta_+}_S)\ket{i'}_{S'}\ge1-\epsilon^2$, and thus we obtain 
 \eq{
 \sum_k|a_k|^2\ge1-\epsilon^2.
 }
 Furthermore, $\sum_kM^\dagger_kM_k=\mathds{1}_S$ implies
 \eq{ 1&=\sum_k\bra{\eta_+}M^\dagger_kM_k\ket{\eta_+}=\sum_k|a_k|^2+\sum_{k}\sum_{m\neq i'}|b_{k,m}|^2,\label{asum1_approx}\\
0&=\sum_k\bra{\eta_-}M^\dagger_kM_k\ket{\eta_+}=\sum_kc^*_ka_k+\sum_{k}\sum_{m\neq i'}d^*_{k,m}b_{k,m},\label{acsum0_approx}\\
0&=\sum_k\bra{l}M^\dagger_kM_k\ket{\eta_+}=\sum_ke^*_{k,l}a_k,\enskip\forall l\mbox{ s.t. }l\neq i,j,\\ 0&=\sum_k\bra{l}M^\dagger_kM_k\ket{\eta_-}=\sum_ke^*_{k,l}c_k,\enskip\forall l\mbox{ s.t. }l\neq i,j.
}
 From these relations, we obtain
\eq{
 \sum_{k}\sum_{m\neq i'}|b_{k,m}|^2&=1-\sum_k|a_k|^2\le\epsilon^2,\\
 |\sum_kc^*_ka_k|&=|\sum_{k}\sum_{m\neq i'}d^*_{k,m}b_{k,m}|\nonumber\\
 &\le\sqrt{\sum_{k}\sum_{m\neq i'}|d_{k,m}|^2}\sqrt{\sum_{k}\sum_{m\neq i'}|b_{k,m}|^2},\nonumber\\
 &\le\epsilon.
 }
 Therefore, we obtain
 \eq{ \tau_{S',i'}=&\bra{i'}_{S'}\tau_{S'}\ket{i'}_{S'}\nonumber\\
 =&\sum_k\bra{i'}_{S'}M_k\tau_S M^\dagger_k\ket{i'}_{S'}\nonumber\\
=&\sum_k|a_k|^2\bra{\eta_+}\tau_S\ket{\eta_+}_S +\sum_k|c_k|^2\bra{\eta_-}\tau_S\ket{\eta_-}_S
+\sum_{k}\sum_{l\neq i,j}|e_{k,l}|^2\bra{l}_{S}\tau_S\ket{l}_{S}\nonumber\\
&+2\mathrm{Re}\left(\sum_kc^*_ka_k\bra{\eta_+}\tau_S\ket{\eta_-}\right)
+2\mathrm{Re}\left(\sum_{l\neq i,j}\sum_ke^*_{k,l}a_k\bra{\eta_+}\tau_S\ket{l}\right)
+2\mathrm{Re}\left(\sum_{l\neq i,j}\sum_ke^*_{k,l}c_k\bra{\eta_-}\tau_S\ket{l}\right)
 \nonumber\\
 =&\sum_k|a_k|^2(r\tau_{S,i}+(1-r)\tau_{S,j})
+\sum_k|c_k|^2((1-r)\tau_{S,i}+r\tau_{S,j})\nonumber\\
 &+\sum_{k}\sum_{l\neq i,j}|e_{k,l}|^2\tau_{S,l} +2\mathrm{Re}\left(\sum_kc^*_ka_k\right)\sqrt{r(1-r)}(\tau_{S,i}-\tau_{S,j})\nonumber\\
\ge&(1-\epsilon^2)\tau_{S',i'}+\sum_k|c_k|^2((1-r)\tau_{S,i}+r\tau_{S,j})
-2\epsilon\sqrt{r(1-r)}(\tau_{S,j}-\tau_{S,i}).\label{eq:c_k_estimate}
}
Because of the definition of $r$, \eqref{eq:c_k_estimate} is converted as follows:
\eq{ \sum_k|c_k|^2&\le\frac{\epsilon^2\tau_{S',i'}+2\epsilon\sqrt{(\tau_{S,j}-\tau_{S',i'})(\tau_{S',i'}-\tau_{S,i})}}{\tau_{S,i}+\tau_{S,j}-\tau_{S',i'}}\nonumber\\
 &=\epsilon f(\epsilon,\tau_{S,i},\tau_{S',i'},\tau_{S,j}).
 }
 where
 \eq{ f(\epsilon,\tau_{S,i},\tau_{S',i'},\tau_{S,j}):=\frac{\epsilon\tau_{S',i'}+2\sqrt{(\tau_{S,j}-\tau_{S',i'})(\tau_{S',i'}-\tau_{S,i})}}{\tau_{S,i}+\tau_{S,j}-\tau_{S',i'}}.
}

Now, let us show that $\Lambda$ satisfies
\eqref{eq:cost_approx_trans_pre2}.
We define $\bbP=\{\dm{\eta_+},\dm{\eta_-}\}$ and a recovery channel
\eq{ \calR(...)=\Tr[\kb{i'_{S'}}...]\kb{\eta_{+}}_S+\Tr[(1-\kb{i'_{S'}})...]\kb{\eta_{-}}_S.
}
Then, we obtain
\eq{
D_F(\calR\circ\Lambda(\eta_+),\eta_+)^2&=1-\bra{i'}_{S'}\Lambda(\eta_+)\ket{i'}_{S'}\nonumber\\
 &=1-\sum_{k}|a_k|^2\nonumber\\
 &\le\epsilon^2,\\ D_F(\calR\circ\Lambda(\eta_-),\eta_-)^2&=1-\Tr[(1-\ket{i'}_{S'}\bra{i'}_{S'})\Lambda(\eta_-)]\nonumber\\
&=\sum_{k}|c_k|^2\nonumber\\
 &\le\epsilon f(\epsilon,\tau_{S,i},\tau_{S',i'},\tau_{S,j})
 }
and
 \eq{
 \delta(\Lambda,\bbP)&\le \sqrt{\frac{D_F(\calR\circ\Lambda(\eta_+),\eta_+)^2+D_F(\calR\circ\Lambda(\eta_-),\eta_-)^2}{2}}\nonumber\\
&\le\sqrt{\frac{\epsilon^2+\epsilon f(\epsilon,\tau_{S,i},\tau_{S',i'},\tau_{S,j})}{2}}.
 }

We can also directly estimate $\calC(\Lambda,\bbP)$ as follows. 
 Note that the operator $\Lambda^\dagger(H_{S'})$ satisfies
 \eq{
|\bra{\eta_+}_{S}\Lambda^\dagger(H_{S'})\ket{\eta_{-}}_S| &=|\sum_{k}\bra{\eta_+}_{S}M^\dagger_kH_{S'}M_k\ket{\eta_{-}}_S|\nonumber\\
&=\left|\sum_{k}\left(a^*_{k}\bra{i'}_{S'}+\sum_{m:m\ne i'}b^*_{k,m}\bra{m}_{S'}\right)H_{S'}\left(c_k\ket{i'}_{S'}+\sum_{m:m\ne i'}d_{k,m}\ket{m}_{S'}\right)\right|
\nonumber\\ &=\left|E_{S',i'}\sum_{k}a^*_kc_k+\sum_{k}\sum_{m:m\ne i'}b^*_{k,m}d_{k,m}E_{S',m}\right|\nonumber\\
&=\left|\sum_{k}\sum_{m:m\ne i'}b^*_{k,m}d_{k,m}(E_{S',m}-E_{S',i'})\right|\nonumber\\
&\le\sqrt{\sum_{k}\sum_{m:m\ne i'}|b_{k,m}|^2}\sqrt{\sum_{k}\sum_{m:m\ne i'}|d_{k,m}|^2(E_{S',m}-E_{S',i'})^2}
\nonumber\\
 &\le\epsilon\Delta(H_{S'}).
}
Therefore,
 \eq{
\calC(\Lambda,\mathbb{P})&=|\bra{\eta_+}(H_{S}-\Lambda^\dagger(H_{S'}))\ket{\eta_-}_S|\nonumber\\
&\ge|\bra{\eta_+}(E_{S,i}\kb{i}_S+E_{S,j}\kb{j}_S)\ket{\eta_-}|-\epsilon\Delta(H_{S'})\nonumber\\
 &=\sqrt{r(1-r)}|E_{S,i}-E_{S,j}|-\epsilon\Delta(H_{S'}).
 }

Applying Theorem~\ref{thm:SIQ} to $\Lambda$ and $\bbP$, we obtain
\eq{
\sqrt{\calF_c(\Lambda)}&\ge\frac{\sqrt{r(1-r)|E_{S,i}-E_{S,j}|}-\epsilon\Delta(H_{S'})}{\sqrt{\frac{\epsilon^2+\epsilon f(\epsilon,\tau_{S,i},\tau_{S',i'},\tau_{S,j})}{2}}}-\Delta(H_{S})-\Delta(H_{S'})\nonumber\\
&\ge\frac{\sqrt{r(1-r)|E_{S,i}-E_{S,j}|}}{\sqrt{\frac{\epsilon^2+\epsilon f(\epsilon,\tau_{S,i},\tau_{S',i'},\tau_{S,j})}{2}}}-\sqrt{2}\Delta(H_{S'})-\Delta(H_{S})-\Delta(H_{S'})\nonumber\\
&\ge\frac{\sqrt{(\tau_{S,j}-\tau_{S',i'})(\tau_{S',i'}-\tau_{S,i})|E_{S,i}-E_{S,j}|}}{(\tau_{S,j}-\tau_{S,i})\sqrt{\epsilon}\sqrt{\frac{\epsilon+f(\epsilon,\tau_{S,i},\tau_{S',i'},\tau_{S,j})}{2}}}-\Delta(H_{S})-3\Delta(H_{S'})\nonumber\\
&\ge\frac{\sqrt{4\kappa(\epsilon,\beta,\tau_{S,i},\tau_{S',i'},\tau_{S,j})}}{\sqrt{\epsilon}}-\Delta(H_{S})-3\Delta(H_{S'}),
}
which concludes the proof. 
Here, we used the following conversion in the last line:
\eq{
\frac{(\tau_{S,j}-\tau_{S',i'})(\tau_{S',i'}-\tau_{S,i})|E_{S,i}-E_{S,j}|}{(\tau_{S,j}-\tau_{S,i})^2\frac{\epsilon+f(\epsilon,\tau_{S,i},\tau_{S',i'},\tau_{S,j})}{2}}
&=\frac{2(\tau_{S,j}-\tau_{S',i'})(\tau_{S',i'}-\tau_{S,i})\log\frac{\tau_{S,j}}{\tau_{S,i}}}{\beta(\tau_{S,j}-\tau_{S,i})^2(\epsilon+f(\epsilon,\tau_{S,i},\tau_{S',i'},\tau_{S,j}))}\nonumber\\
&=\frac{2(\tau_{S,j}-\tau_{S',i'})(\tau_{S',i'}-\tau_{S,i})\log\frac{\tau_{S,j}}{\tau_{S,i}}}{\beta(\tau_{S,j}-\tau_{S,i})^2(\epsilon+\frac{\epsilon\tau_{S',i'}+2\sqrt{(\tau_{S,j}-\tau_{S',i'})(\tau_{S',i'}-\tau_{S,i})}}{\tau_{S,i}+\tau_{S,j}-\tau_{S',i'}})}\nonumber\\
&=\frac{2(\tau_{S,i}+\tau_{S,j}-\tau_{S',i'})(\tau_{S,j}-\tau_{S',i'})(\tau_{S',i'}-\tau_{S,i})\log\frac{\tau_{S,j}}{\tau_{S,i}}}{\beta(\tau_{S,j}-\tau_{S,i})^2(\epsilon(\tau_{S,i}+\tau_{S,j})+2\sqrt{(\tau_{S,j}-\tau_{S',i'})(\tau_{S',i'}-\tau_{S,i})})}\nonumber\\
&\ge\frac{2(\tau_{S,i}+\tau_{S,j}-\tau_{S',i'})(\tau_{S,j}-\tau_{S',i'})(\tau_{S',i'}-\tau_{S,i})\log\frac{\tau_{S,j}}{\tau_{S,i}}}{\beta(\tau_{S,j}-\tau_{S,i})^2(\tau_{S,i}+\tau_{S,j}+2\sqrt{(\tau_{S,j}-\tau_{S',i'})(\tau_{S',i'}-\tau_{S,i})})}\nonumber\\
&=4\kappa(\beta,\tau_{S,i},\tau_{S',i'},\tau_{S,j}).
}
}
\end{proof}

\section{Upper bounds for the coherence cost of general quantum operations and pairwise reversible Gibbs preserving operations (Proofs of Theorems~\ref{thm:tight bound}~and~\ref{thm:general upper})}\label{app:upper bounds}

We first show Theorem~\ref{thm:general upper}, which we later utilize to prove Theorem~\ref{thm:tight bound}.

\begin{theorem}[Theorem~\ref{thm:general upper} in the main text]\label{thm:general upper app}
Let $\Lambda:S\to S'$ be an arbitrary quantum channel admitting a dilation form
\bal
 \Lambda(\rho) = \Tr_{E'}\left(V(\rho\otimes \dm{\eta})V^\dagger\right)
\eal
for some environments $E$ and $E'$ such that $S\otimes E = S'\otimes E'$, some unitary $V$ on $S\otimes E$, and some pure incoherent state $\ket{\eta}$ on $E$.
Then, 
\bal
\sqrt{\calF_c^\epsilon(\Lambda)}\leq \frac{\Delta(H_{\rm tot}-V^\dagger H_{\rm tot}V)}{2\epsilon} + \sqrt{2}\Delta(H_{\rm tot})
\label{eq:coherence cost upper bound general}
\eal
where $H_{\rm tot}=H_S\otimes\mathds{1}_E + \mathds{1}_S\otimes H_E=H_{S'}\otimes\mathds{1}_{E'} + \mathds{1}_{S'}\otimes H_{E'}$, and $\Delta(O)$ is the difference between the minimum and maximum eigenvalues of an operator $O$.
\end{theorem}

\begin{proof}
Ref.~\cite[Theorem 2]{Tajima2020coherence} shows that an arbitrary unitary channel $\calV(\cdot)=V\cdot V$ on a system with Hamiltonian $H$ can be implemented with error $\epsilon$ with coherence cost $\frac{\Delta(H-V^\dagger H V)}{2\epsilon}+\sqrt{2}\Delta(H)$.
Let $\calV_\epsilon$ be such a channel approximating $\calV$ satisfying $D_F(\calV_\epsilon,\calV)\leq \epsilon$, and define $\Lambda_\epsilon\coloneqq \Tr_{E'}\circ \calV_\epsilon\circ \calP_{\ket{\eta}}$ where $\calP_{\ket{\eta}}(\rho)=\rho\otimes \dm{\eta}$ is a state preparation channel. 
Noting that $\calP_{\ket{\eta}}$ and $\Tr_{E'}$ can be implemented with no coherence cost, we get 
\bal
\sqrt{\calF_c(\Lambda_\epsilon)} \leq  \sqrt{\calF_c(\calV_\epsilon)} \leq \frac{\Delta(H_{\rm tot}-V^\dagger H_{\rm tot}V)}{2\epsilon} + \sqrt{2}\Delta(H_{\rm tot}).
\eal

Therefore, it suffices to show that $D_F(\Lambda_\epsilon,\Lambda)\leq \epsilon$, which would ensure that $\sqrt{\calF_c^\epsilon(\Lambda)}\leq \sqrt{\calF_c(\Lambda_\epsilon)}$ and result in the advertised upper bound.
This can indeed be checked by
\bal
 D_F(\Lambda_\epsilon,\Lambda) &= D_F(\Tr_{E'}\circ\calV_\epsilon\circ\calP_{\ket{\eta}},\Tr_{E'}\circ\calV\circ\calP_{\ket{\eta}})\\
 &\leq D_F(\Tr_{E'}\circ\calV_\epsilon,\Tr_{E'}\circ\calV)\\
 &\leq D_F(\calV_\epsilon,\calV)\\
 &\leq \epsilon
\eal
where the second line comes from the definition of the channel purified distance \eqref{eq:purified channel distance definition app}, the third line is because of the data-processing inequality of the purified distance, and the fourth line follows from the assumption of $\calV_\epsilon$.  
\end{proof}

We next show Theorem~\ref{thm:tight bound}, which shows that the lower bound in Theorem~\ref{thm:lower bound} is almost tight, where $\calC(\Lambda,\bbP)$ serves as a fundamental quantity that characterizes the coherence cost for a Gibbs-preserving Operation.

\begin{theorem}[Theorem \ref{thm:tight bound} in the main text]\label{thm:tight bound_SM}
For every real number $a>0$, there is a pairwise reversible Gibbs-preserving Operation $\Lambda$ and a reversible pair $\bbP$ such that $\calC(\Lambda,\bbP)>0$ and 
\bal
\frac{\calC(\Lambda,\bbP)}{\epsilon}-a\leq\sqrt{\calF_c^\epsilon(\Lambda)}\le\frac{\sqrt{2}\calC(\Lambda,\bbP)}{\epsilon}+a.
\label{eq:tight bound app}
\eal
\end{theorem}
\begin{proof}
For a given $a>0$, consider a qubit system $S$ with Hamiltonian $H_S= \tilde a \dm{1}$ with $\tilde a=a/\sqrt{2}$ and another qubit system $S'$ with trivial Hamiltonian $H_{S'}=0$.
Consider a channel $\Lambda:S\to S'$ defined by 
\bal
 \Lambda(\rho) = \brakets{+}{\rho}{+}\dm{0} + \brakets{-}{\rho}{-}\dm{1}.
\eal
This is evidently a Gibbs-preserving Operation, noting that $H_{S'}=0$.
Take the state pair $\bbP=\{\dm{+},\dm{-}\}$. 
Since both states are reversible under $\Lambda$, we have $\delta(\Lambda,\bbP)=0$, and $\calC(\Lambda,\bbP)>0$ can be checked by direct computation.  
This ensures that Theorem~\ref{thm:lower bound app} can be applied, and the lower bound in \eqref{eq:tight bound app} then immediately follows noting that $\Delta(H_S)=\tilde a\leq a$ and $H_{S'}=0$.

To get the upper bound, notice that the channel $\Lambda$ can be implemented by 
\bal
\Lambda(\rho)=\Tr_{S}(V\rho\otimes\dm{0}_{S'}V^\dagger)
\eal
where $V={\rm CNOT}_{SS'}\,U_H\otimes\mathds{1}_{S'}$ is a unitary on $SS'$ where $U_H$ is the Hadamard gate and ${\rm CNOT}_{SS'}$ is the CNOT gate controlled on $S$.
Noting that $H_{S'}=0$ and thus $H_{\tot}=H_S\otimes\mathds{1}_{S'}$, we get 
\bal
H_{\tot}-V^\dagger H_{\tot} V = H_S\otimes\mathds{1} - V^\dagger (H_S\otimes \mathds{1}) V = (H_S - U_H H_S U_H)\otimes \mathds{1}.
\eal
This gives 
\bal
 \Delta(H_{\tot}- V^\dagger H_{\tot}V) = \Delta(H_S-U_H H_S U_H) = \sqrt{2} \tilde a.
 \label{eq:spectrum gap}
\eal
On the other hand, 
\bal
 \calC(\Lambda,\bbP)=\left|\brakets{+}{H_S-U_H H_S U_H}{-}\right| = \frac{\tilde a}{2}.
 \label{eq:C value}
\eal
Combining \eqref{eq:spectrum gap} and \eqref{eq:C value} gives 
\bal
 \Delta(H_{\tot}-V^\dagger H_{\tot} V) = 2\sqrt{2}\,\calC(\Lambda,\bbP),
\eal
from which the upper bound in \eqref{eq:tight bound app} follows by using Theorem~\ref{thm:general upper app} and noting $\Delta(H_S)=\tilde a=a/\sqrt{2}$ and $\Delta(H_{S'})=0$.
\end{proof}

\end{document}